\documentclass[journal,twoside]{IEEEtran}

\IEEEoverridecommandlockouts

\usepackage{algorithmicx}
\usepackage{algorithm}
\usepackage{algpseudocode}

\usepackage{amsmath,amssymb,amsfonts,amsthm}
\usepackage[table]{xcolor}
\usepackage{graphicx}
\usepackage{cite}
\usepackage{color}
\usepackage{multirow}
\usepackage{mathtools}
\usepackage{tikz,amsmath}
\usetikzlibrary{shapes.geometric,positioning,arrows.meta}
\usepackage{siunitx}
\usepackage{balance}
\usetikzlibrary{shapes.geometric}
\usepackage{soul}
\usepackage{comment}
\usepackage{chngcntr}
\usepackage{enumitem,kantlipsum}
\usepackage{pgfplots}
\newtheorem{theorem}{Theorem}

\newtheorem{remark}{Remark}

\pgfplotsset{compat=1.18} 

\setulcolor{orange}

\setstcolor{red}

\begin{document}
\title{\vspace{-0.07in}Optimal Steady-State Secondary Control of\\MT-HVdc Grids with Reduced Communications}
\author{Babak Abdolmaleki and Gilbert Bergna-Diaz 
\vspace{-0.15in}


\thanks{This work was supported by the Department of Electric Energy, Norwegian University of Science and Technology (NTNU), under Grant 988775100. Paper no. XXXX-X-2023. (Corresponding author: Babak Abdolmaleki.)}
\thanks{The authors are with the Department of Electric Energy, Norwegian University of Science and Technology (NTNU), 7491 Trondheim, Norway (e-mail: babak.abdolmaleki@ntnu.no; gilbert.bergna@ntnu.no).}
}
\markboth{Submitted for Publication}{Submitted for Publication}

\maketitle
\begin{abstract}
In this paper, we propose a centralized secondary control for the real-time steady-state optimization of multi-terminal HVdc grids under voltage and current limits. First, we present the dynamic models of the grid components, including the modular multilevel converter (MMC) stations and their different control layers. We also derive the quasi-static input-output model of the system, which is suitable for the steady-state control design. Second, we formulate a general optimization problem using this quasi-static model and find the Karush-Kuhn-Tucker optimality conditions of its solutions. Third, we propose a secondary control based on primal-dual dynamics to adjust the voltage setpoints of the dispatchable MMCs, with which the system asymptotically converges to a steady state that satisfies these optimality conditions. Fourth, we provide a communication triggering mechanism to reduce the communication traffic between the secondary control unit and the MMC stations. Finally, we verify our proposal for different case studies by adapting it to an offshore multi-terminal HVdc grid composed of heterogeneous MMC stations simulated in the MATLAB/Simulink environment. The problems of proportional current minimization and loss reduction are two special case studies. 
\end{abstract}
\begin{IEEEkeywords}
HVdc grid, optimization, voltage control.\vspace{-0.07in}
\end{IEEEkeywords}
\section{Introduction}
\subsection{Background and Motivation}
\IEEEPARstart{M}{ulti-terminal} (MT) dc grids are getting increasingly more applications in modern power systems ranging from low-voltage microgrids to inter-area high-voltage dc (HVdc) grids.
They are, for instance, important infrastructures for integration of offshore wind power into the European energy systems \cite{Misyris2022,ENTSO-E}. Multi-purpose MT-HVdc grids are also emerging as viable solutions to interconnect several markets and offshore power generation and consumption hubs (e.g., oil and gas platforms, electrolyzers, etc.) \cite{ENTSO-E}.

In a MT-HVdc grid, each terminal is connected to a converter station, commonly based on a Modular Multilevel Converter (MMC) \cite{Ludois2014}.
Depending on their system-level functionality, the MMC stations can be governed by different lower and upper level control loops \cite{CIGRE,Dominic2022,Soler2023}. 
They can be categorized into dc grid-forming (dc-GFM) and dc grid-following (dc-GFL) stations \cite{Dominic2022,Soler2023}.
A set of dc-GFM stations are in charge of shaping (forming) the HVdc grid voltage profile while the dc-GFL MMCs are tasked with delivering power to/from the ac side. It should be noted that both the dc-GFM and dc-GFL stations can also operate in ac grid-forming (ac-GFM) mode and participate in forming of the ac side voltage and frequency, see, e.g., \cite{Dominic2022,Soler2023}. 

It is well known that, in steady state, dc grids can be represented by a resistive network where the distribution of the branch currents depends on the node voltage differences \cite{Lingwen2014,Kjetil2012,Beerten2013,Abdelwahed2017,Jabr2021}. This means that the voltage setpoint of dc-GFM stations can be adjusted to achieve a desired load (current or power) flow in MT-HVdc grids. Some practical examples of optimization problems in this regard are economic dispatch and proportional load minimization for dc-GFM stations, as well as grid voltage profile improvement and loss reduction.

Indeed, a simple control solution to adjust the voltage setpoint of the dc-GFM stations is the well-known droop control \cite{Kjetil2012,Beerten2013}.
With an optimal design and selection of its parameters, droop control can provide the desired load flow in dc grids for a nominal model under a generation-consumption scenario -- see for example the work in \cite{Abdelwahed2017,Jabr2021,Spiros2023}.
However, in a MT-HVdc grid with high uncertainty in generation, consumption, and constraints, even a well-designed droop controller may not provide optimality for all scenarios. This problem, together with the possibility that the grid operator may require different optimization objectives over different time periods, necessitates real-time voltage adjustments based on more than just local information. Under the hierarchical control policy in HVdc grids, this lies in the secondary and tertiary control levels, which generally operate on a longer (slower) processing timescale than the lower levels \cite{CIGRE,Alvarez2015}. The secondary controller is responsible for restoring the operating conditions to pre-disturbance levels by, e.g., updating the voltage setpoints of the dc-GFM stations to follow a power or current setpoint given by the tertiary control. The tertiary controller, typically based on optimal power flow, generates the optimal setpoints for the secondary controller \cite{Alvarez2015}. The objective of this paper is to design a secondary controller for dc-GFM stations that directly adjusts their voltage setpoints to provide optimal steady-state operation of MT-HVdc grids in real time.\vspace{-0.07in}

\subsection{Literature Review and Research Gaps}
Optimal steady-state operation of MT-HVdc grids has been studied in many papers, and various techniques have been proposed to achieve different useful objectives.
In some works, generally speaking, an offline optimization, repeatedly solved over long periods of time, is used to update the local primary controllers of the converter stations.
For example, in \cite{Gavriluta2015}, a secondary voltage control is proposed for network loss minimization. In this method, an offline optimization problem is solved and its solutions are sent to the converter stations.
In \cite{Papangelis2017}, a centralized model predictive secondary control is proposed for converter stations to follow their given power references subject to tight regulation of the average voltage and the operational limits.
In \cite{CarmonSanchez2020}, a model predictive secondary control is proposed to support the droop controller. The control objectives considered in this work are voltage and power tracking.
A generalized droop control is proposed in \cite{Eriksson2018} to improve the functionality of the converters, especially in case of failure of one station, by better distributing the load among the remaining stations. In this method, the grid operator performs an offline constrained optimization and computes the droop gain matrix accordingly. A somewhat similar approach, but with an adaptive local droop controller, is proposed in \cite{Yogarathinam2019}, where stability constraints are also taken into account.
In \cite{Shinoda2022}, an adaptive droop control is proposed whose parameters are optimally selected based on the available headroom and voltage containment reserve of the converter stations.
In \cite{Zhang2021,Zhang2022}, some secondary controllers are proposed for load sharing, voltage regulation, and loss minimization. In these works, an offline optimization problem is solved and used to adjust the power or voltage references.
A cost-based adaptive droop control is proposed in \cite{Song2021}, where an area cost minimization problem is solved every five minutes and the droop parameters are updated accordingly.
In \cite{Xie2023}, a secondary control is proposed that provides average voltage regulation and loss minimization by repeatedly solving an offline optimization problem and updating the primary controllers.
A somewhat similar secondary control is proposed in \cite{Li2018} for loss minimization and load sharing in dc grids.

In another line of research, innovative dynamic controllers, usually based on consensus algorithms and multi-agent systems, are proposed to achieve some specific control objectives.
In \cite{Wang2021}, a cooperative control is proposed for frequency support and power sharing among the converters, where the dc voltage control is not studied.
In \cite{Aram2018}, a distributed secondary control is proposed to provide accurate load sharing among the converters.
To provide load sharing and voltage regulation, a secondary control is proposed in \cite{Aram2023}, where three different loops are designed to provide frequency support, power sharing, and voltage regulation by adjusting the voltage setpoints of the stations. 
A secondary control is proposed in \cite{Wang2020}, which realizes an adjustable trade-off between load sharing and voltage regulation objectives.
In \cite{Lotfifard2022}, a distributed adaptive control based on a dynamic consensus algorithm is proposed for power sharing among the converters.
In \cite{Zhang2020, Yang2022}, some distributed controllers are proposed to achieve proportional power sharing and average voltage regulation.

The methods in \cite{Gavriluta2015,Papangelis2017,CarmonSanchez2020,Eriksson2018,Yogarathinam2019,Shinoda2022,Zhang2021,Zhang2022,Song2021,Xie2023, Li2018} are all based on offline decision making. In principle, they are feedforward control and optimization methods, which are not robust to unknown disturbances and model uncertainties.
Unlike the work in \cite{Gavriluta2015,Papangelis2017,CarmonSanchez2020,Eriksson2018,Yogarathinam2019,Shinoda2022,Zhang2021,Zhang2022,Song2021,Xie2023, Li2018}, the methods in \cite{Wang2021,Aram2018,Aram2023,Wang2020,Lotfifard2022,Zhang2020,Yang2022} are online feedback controllers.
However, despite their robust performance, they may not control the system to the best steady state for which they were designed, especially when operating constraints are considered.
To address this problem, which is common to many engineering disciplines, there has been a recent tendency to incorporate the relevant constrained optimization algorithms into the feedback control loops in a systematic way \cite{Krishnamoorthy2022,Emiliano2018,Ortman2020,Ortman2022,Häberle2021,Hauswirth2021}.
In this approach, the feedback control laws steer the system towards an optimal and admissible steady state in real time, while preserving some robustness guarantees.
This core idea has gained attention in various disciplines such as process control \cite{Krishnamoorthy2022}, power systems \cite{Emiliano2018,Ortman2020,Ortman2022}, and control systems \cite{Häberle2021,Hauswirth2021}, where it is generally known as Real-Time Optimization (RTO) or Online Feedback Optimization (OFO).
In this study, we apply this idea to MMC-based MT-HVdc grids and provide a framework to bridge the gaps in the reviewed literature.

\subsection{Contributions}
In this paper, we propose a feedback control technique based on \emph{partial} model information and for a general constrained optimization problem, resulting in an optimal steady-state secondary voltage control.
The following contributions are made in this paper.
\begin{itemize}[leftmargin=*]
    \item We consider a MT-HVdc grid interconnecting several areas and offshore production-consumption hubs. We first present the dynamic model of the grid components, including transmission lines and MMC stations based on dc-GFM and ac-GFM technologies. We then derive the linearized quasi-static input-output relationship of the system, which is suitable for the steady-state control design.
    \item We formulate a feedback-based optimization problem and find the Karush-Kuhn-Tucker (KKT) optimality conditions of its solution. We then propose a feedback controller that asymptotically steers the system to a steady state where these conditions are satisfied. The proposed technique is a secondary controller for the dc-GFM stations and is based on partial model information, i.e., it requires only the sensitivity matrix in the quasi-static model. We construct this matrix using the network conductance matrix and the measurements from the ac-GFM stations.
    \item In order to reduce the communication traffic between the control center and the MMC stations, we also offer a communication triggering mechanism based on the event-triggered control technique \cite{Abdolmaleki2019}.
    \item We demonstrate the effectiveness of our proposal by performing a Lyapunov stability analysis and MATLAB-based numerical simulations for different generation-consumption scenarios and optimization objective functions, including proportional current minimization and loss reduction.
\end{itemize}

The rest of this paper is structured as follows. In Section~\ref{Sec:Model}, we present the dynamic and quasi-static system models, which are suitable for time-domain simulations and secondary control design, respectively. In Section~\ref{Sec:Optimization}, we formulate the optimization problem and find the optimality conditions of its solutions. We present our controller in Section~\ref{Sec:ProposedController}, where the stability analysis and the communication reduction mechanism are also presented. In Section~\ref{Sec:CaseStudies}, we present the simulation results for different case studies. Finally, Section~\ref{Sec:Conclusion} concludes the paper.
\section{Dynamic and Quasi-Static System Modelling}
\label{Sec:Model}
In this section, we introduce the system modeling used in our study. First, we present the dynamic models used in the simulations, and then we derive the quasi-static system model suitable for secondary control design and study.
\subsection{System Dynamic Model}
In a dc grid, depending on the application, voltage level, grid topology, and vendor-specific technology in use, different converters with different low-level inner control loops may coexist. In our case studies, we consider a MT-HVdc grid based on MMCs -- see Fig.~\ref{Fig:HVdcGrid}. In the following, we present the dynamics of the dc transmission lines as well as the dynamic models of the MMC stations and the standard cascaded control loops we used to control them. However, our design and analysis can be applied to other types of converters and controllers, depending on the application, technology, and setting of interest.
\begin{figure}
    \centering
    \includegraphics[width=\columnwidth]{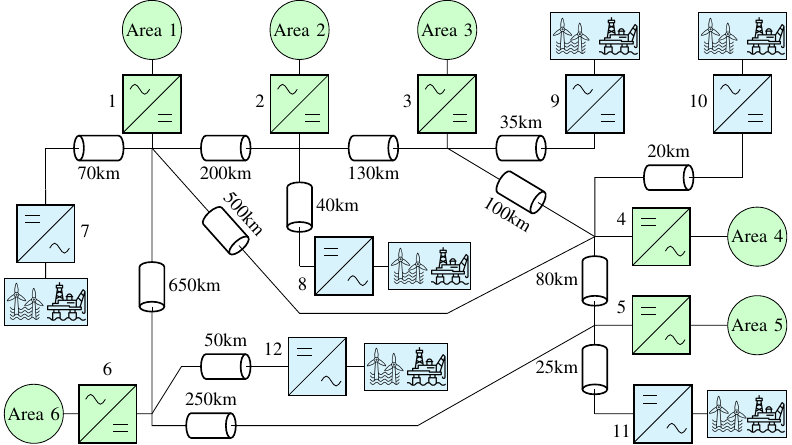}
    \caption{The test MT-HVdc grid based on MMC.}
    \label{Fig:HVdcGrid}
\end{figure}
\subsubsection{Transmission Line Dynamics}
We make use of the fifth-order cable model for each transmission line between the MMC stations in Fig.~\ref{Fig:HVdcGrid}. This model, shown in Fig.~\ref{Fig:Cable}, captures the system level dynamic behavior of the transmission lines with an acceptable accuracy \cite{Sanchez2020}.
\begin{figure}
    \centering
    \includegraphics[width=0.85\columnwidth]{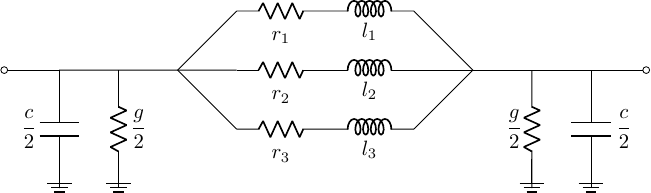}
    \caption{Cable model for dc transmission lines.}
    \label{Fig:Cable}
\end{figure}

\subsubsection{MMC Dynamics}
In \cite{Gilbert2015,Gilbert2018}, under the \textit{compensated modulation} strategy, an energy-based simplified \emph{macroscopic} representation of the MMC is developed which maintains the accuracy of the converter dynamic behavior between its ac and dc terminals.
\begin{figure}
    \centering
    \includegraphics[width=0.95\columnwidth]{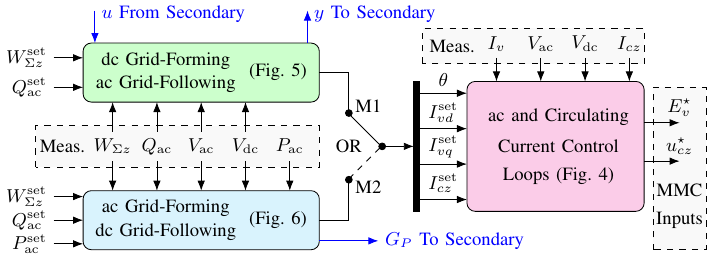}
    \includegraphics[width=0.95\columnwidth]{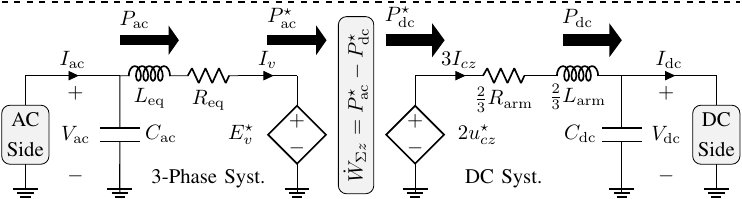}
    \caption{System-level dynamic model of an MMC \cite{Gilbert2018} and its control system.}
    \label{Fig:MMC}
\end{figure}
As shown in Fig.~\ref{Fig:MMC}, in this model the dc and ac sides of the MMC are coupled through the dynamics of zero sequence energy sum $W_{\Sigma z}$. The parameters, $R_{\rm arm}$ and $L_{\rm arm}$ are the MMC arm resistance and inductance, while $R_{\rm eq}=\tfrac{1}{2}R_{\rm arm}+R_f$ and $L_{\rm eq}=\tfrac{1}{2}L_{\rm arm}+L_f$ are its \textit{equivalent} output resistance and inductance on the ac side, with $R_f$ and $L_f$ the actual ac-side resistance and inductance, respectively. The capacitors $C_{\rm dc}$ and $C_{\rm ac}$ represent the total capacitors at the points of coupling to the dc and ac grids, including the line capacitors and the dedicated physical capacitors, if applicable. The dc input $u_{cz}^*$ is the zero-sequence component of the leg voltage, which drives the MMC circulating current with the zero-sequence component $I_{cz}$. The ac input $E_v^*$ is the three-phase voltage at the MMC terminal, which drives the ac current $I_v$. The zero-sequence energy dynamics, couples the dynamics of the ac and dc sides, where $P_{\rm dc}^*=6u_{cz}^*I_{cz}$ and $P_{\rm ac}^* = E_{v}^* \cdot I_v$ (dot-product of the three-phase signals). The MMC station absorbs the power $P_{\rm ac}=V_{\rm ac}\cdot I_v$ from the ac side and injects the power $P_{\rm dc}=3V_{\rm dc}I_{cz}$ to the dc side. In our simulations, we make use of this model, which is shown to be efficient, albeit accurate, for power system-oriented studies and simulations\cite{Gilbert2018}.
\subsubsection{MMC Control System}
Fig.~\ref{Fig:MMC} shows the standard hierarchical control structure of the MMC stations. The dc and ac inputs $u_{cz}^*$ and $E_{v}^*$ are generated by the circulating and ac current controllers, respectively.
As depicted in Fig.~\ref{Fig:CurrentControl}, these control loops synthesize the feedback measurements in a synchronous ($dq$) reference frame that rotates with the speed of $\dot{\theta}$, and they are designed based on the standard proportional-integral (PI) controllers \cite{Yazdani2010,Dominic2022,Soler2023}.
\begin{figure}
    \centering
    \includegraphics[width=.85\columnwidth]{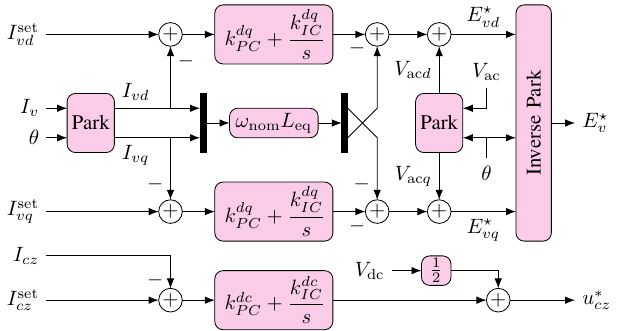}
    \caption{The inner ac and zero-sequence circulating current control loops.}
    \label{Fig:CurrentControl}
\end{figure}
As shown in Fig.~\ref{Fig:MMC}, the current setpoints and the rotating angle are generated by either the ac grid-forming (ac-GFM) controller in Fig.~\ref{Fig:ac-GFM} or the dc grid-forming (dc-GFM) controller in Fig.~\ref{Fig:dc-GFM}, depending on the system-level functionality of the MMC station. In the MT-HVdc grid shown in Fig.~\ref{Fig:HVdcGrid}, the MMCs connected to the areas are governed by dc-GFM technique while the offshore MMC stations operate in ac-GFM mode.
\begin{figure}
    \centering
    \includegraphics[width=0.9\columnwidth]{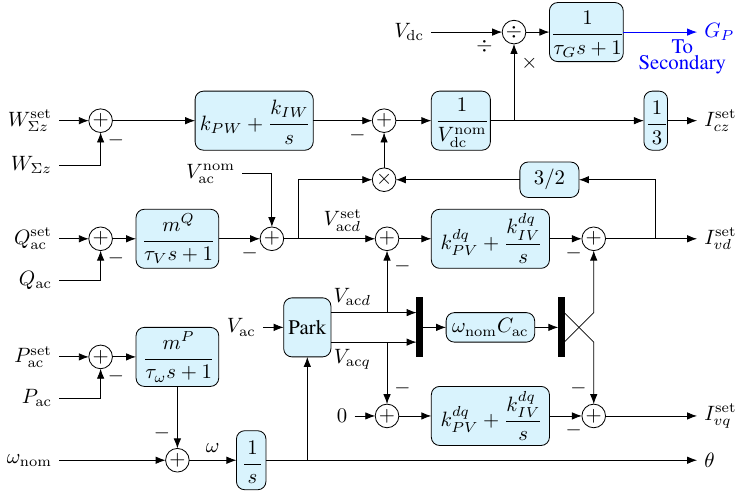}
    \caption{The ac grid-forming/dc grid-following (ac-GFM) outer control loops. This controller is applied to the offshore-connected MMC stations in Fig.~\ref{Fig:HVdcGrid}}
    \label{Fig:ac-GFM}
\end{figure}
\begin{figure}
    \centering
    \includegraphics[width=0.9\columnwidth]{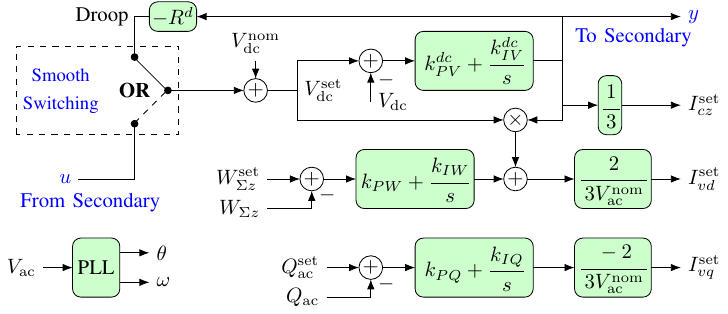}
    \caption{The dc grid-forming/ac grid-following (dc-GFM) outer control loops. This controller is applied to the area-connected MMC stations in Fig.~\ref{Fig:HVdcGrid}.}
    \label{Fig:dc-GFM}
\end{figure}
The main difference between the two technologies is that, \textit{from a dc grid perspective}, the dc-GFM stations are dispatchable and they can be used to control the voltage profile of the HVdc grid, while the ac-GFM stations are constant power sources or loads with grid-dictated dc voltages.

\subsection{Quasi-Static Input-Output Model of MT-HVdc Grid}
In this section, we derive a quasi-static model of the grid that is suitable for steady-state optimization of the system and secondary control design. Assuming that the capacitors in Fig.~\ref{Fig:Cable} are open and the inductors are short-circuited, we can write the dc network current flow equations as \cite{Lingwen2014,Kjetil2012,Beerten2013,Abdelwahed2017,Jabr2021}
\begin{IEEEeqnarray}{c}
\label{eq:QSModel}
I_{{\rm dc}i}=G_iV_{{\rm dc}i} + {\sum}_{j} (1/R_{ij})(V_{{\rm dc}i}-V_{{\rm dc}j}),\quad \forall i,\quad\IEEEyesnumber\IEEEyessubnumber\label{eq:CurrentFlow}
\end{IEEEeqnarray}
where $G_i$ is the total parallel conductance of all the cables connected to the station $i$ and $R_{ij}$ is the total cable resistance between stations $i$ and $j$.
Let $\mathcal{N}$ and $\mathcal{M}$, with the cardinalities $n$ and $m$, denote the sets of dc-GFM and ac-GFM MMC stations, respectively. If the system and the mentioned controllers are properly designed, in steady state the MMCs achieve
\begin{IEEEeqnarray}{rCll}
V_{{\rm dc}i}&=&V_{\rm nom}+u_i,\qquad y_i=I_{{\rm dc}i},&\qquad \forall i\in \mathcal{N},\IEEEyessubnumber\label{eq:DCGMF-QS}\\
I_{{\rm dc}i}&=&P_{{\rm dc}i}/V_{{\rm dc}i},&\qquad \forall i\in \mathcal{M},\IEEEyessubnumber\label{eq:ACGFM-QS}
\end{IEEEeqnarray}
where $(u_i,y_i)$ is the secondary control input-output pair of the $i$th dc-GFM MMC, with $u_i$ to be designed later.
For simplicity, we can linearize the equation \eqref{eq:ACGFM-QS} about $\bar{V}_{{\rm dc}i}$ as
\begin{IEEEeqnarray}{rCll}
I_{{\rm dc}i}&=&-G_{Pi}V_{{\rm dc}i} + \bar{I}_{{\rm dc}i},&\quad \forall i\in \mathcal{M},\label{eq:ACGFM-QS-Linear}
\end{IEEEeqnarray}
where $G_{Pi}=P_{{\rm dc}i}/\bar{V}^2_{{\rm dc}i}$ and $\bar{I}_{{\rm dc}i}=2 P_{{\rm dc}i}/\bar{V}_{{\rm dc}i}$.
Fig.~\ref{Fig:QSModel} shows the circuit diagram of the quasi-static model in \eqref{eq:QSModel} and linearization of the ac-GFM MMCs.
\begin{figure}
    \centering
    \includegraphics[width=\columnwidth]{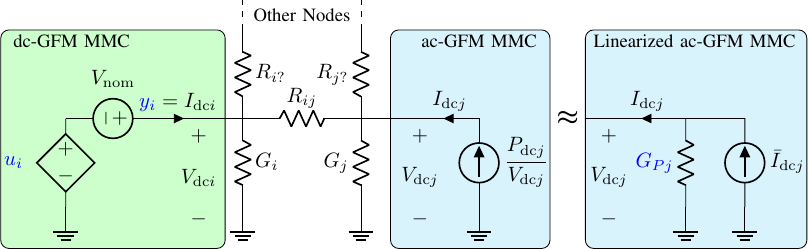}
    \caption{Circuit diagram of the MT-HVdc grid quasi-static model \eqref{eq:QSModel}-\eqref{eq:ACGFM-QS-Linear}.}
    \label{Fig:QSModel}
\end{figure}

For dc-GFM MMCs, let us now define the vectors $u=\mathrm{col}(u_1,\ldots,u_n)$, $y=\mathrm{col}(y_1,\ldots,y_n)$ and $1_n=\mathrm{col}(1,\ldots,1)$. We also define the matrix $G_P=\mathrm{diag}(G_{P1},\ldots,G_{Pm})$ and the vectors $I_\mathcal{M}=\mathrm{col}(\bar{I}_{{\rm dc}1},\ldots,\bar{I}_{{\rm dc}m})$ and $V_\mathcal{M}=\mathrm{col}(V_{{\rm dc}1},\ldots,V_{{\rm dc}m})$ for the ac-GFM MMCs. With this, we can compactly write \eqref{eq:CurrentFlow}, \eqref{eq:DCGMF-QS} and \eqref{eq:ACGFM-QS-Linear} as
\begin{IEEEeqnarray}{rCl}
\begin{bmatrix}
    y\\
    I_\mathcal{M}
\end{bmatrix}&=&
\begin{bmatrix}
    G_\mathcal{N} & G_\mathcal{NM}\\
    G_\mathcal{MN} & G_\mathcal{M}+G_P
\end{bmatrix}
\begin{bmatrix}
    V_{\rm nom}1_n + u \\
    V_\mathcal{M}
\end{bmatrix},
\end{IEEEeqnarray}
where $G_\mathcal{N}$, $G_\mathcal{NM}$, $G_\mathcal{MN}$, are $G_\mathcal{M}$ blocks of the network conductance matrix, obtained from the current flow equations \eqref{eq:CurrentFlow}.
One can now apply Kron reduction\cite{DorflerKR2013} and derive the reduced quasi-static input-output map
\begin{IEEEeqnarray}{rCl}
y &=& G_{\rm red} u + w,\label{eq:Input-OutputMap}
\end{IEEEeqnarray}
where $w = V_{\rm nom} G_{\rm red} 1_n + G_\mathcal{NM}^{\rm red} I_\mathcal{M}$ with the reduced matrices $G_\mathcal{NM}^{\rm red}=G_\mathcal{NM} (G_P +G_\mathcal{M})^\dagger$ and $G_{\rm red}=G_\mathcal{N} - G_\mathcal{NM}^{\rm red} G_\mathcal{MN}$; the superscript $\dagger$ stands for the Moore–Penrose inverse. In the next section, we use this model to design a secondary control for real-time steady-state optimization of the MT-HVdc grid.


\section{Problem Statement and Optimality Conditions}
\label{Sec:Optimization}
In the previous section, we developed the linear quasi-static input-output model \eqref{eq:Input-OutputMap}, where the input and output correspond to the dc current and voltage setpoints as $ I_{{\rm dc}i} \equiv y_i $ and $V_{{\rm dc}i}^{\rm set} \equiv V_{\rm nom} + u_i$, respectively. In this section, we use that model to state the problem and find conditions for optimality of its solution.
Consider the optimization problem
\begin{IEEEeqnarray}{c}
\label{eq:Feedback-based_Program}
\min_{u}\,f(u,y),\IEEEyesnumber\IEEEyessubnumber\\
\text{subject to }
\begin{cases}
y=G_{\rm red} u+w\\
I_{{\rm dc}i}^{\rm min} \leq y_i \leq I_{{\rm dc}i}^{\rm max}\\
V_{{\rm dc} i}^{\rm min} \leq V_{\rm nom} + u_i \leq V_{{\rm dc} i}^{\rm max}\\
\end{cases}\!\!\!\forall i \in \mathcal{N},\quad\IEEEyessubnumber\label{eq:Feedback-based_Constraints}
\end{IEEEeqnarray}
where $f$ is a user-defined convex cost function, and where \eqref{eq:Feedback-based_Constraints} shows the equality and inequality constraints on the inputs and outputs. The above cost function is very general and can be defined to solve many problems such as transmission loss reduction, economic dispatch, proportional current minimization for the dispatchable MMCs, voltage regulation, etc., as we will show in our simulation case studies.
Setting $y=G_{\rm red} u+w$ into $f(u,y)$, one obtains $g(u)=f(u,G_{\rm red} u+w)$ and can formulate the equivalent problem
\begin{IEEEeqnarray}{c}
\label{eq:Model-based_Program}
\min_{u}\,g(u),\IEEEyesnumber\IEEEyessubnumber\\
\text{subject to }
\begin{cases}
I_{\rm dc}^{\rm min} \leq G_{\rm red} u+w \leq I_{\rm dc}^{\rm max}\\
V_{{\rm dc}}^{\rm min} \leq V_{\rm nom}1_n + u \leq V_{{\rm dc}}^{\rm max}
\end{cases}\:\:\:\:\IEEEyessubnumber\label{eq:Model-based_Constraints}
\end{IEEEeqnarray}
where the lower and upper bounds are now in vector form.
Solving \eqref{eq:Model-based_Program} using the available solvers is computationally efficient and can be done repeatedly in real time. However, it requires knowledge of both $G_{\rm red}$ and $w$, which may be subject to model uncertainties and unknown disturbances. This leads to a model mismatch that can degrade the optimal operation of the system and even prevent its safe operation by violating the constraints. In this paper, however, instead of repeatedly solving an optimization problem for different $G_{\rm red}$ and $w$, we take advantage of the output feedback $y$ and propose an online feedback controller that drives the system to the optimal steady state while increasing the robustness guarantees and respecting the operational limits. To this aim, we need to characterize the optimal steady state of the system and then incorporate an optimization algorithm into the feedback controller. We first find the sufficient and necessary conditions for optimality and then design a controller that establishes these conditions.

Using the dual variables $\lambda_{\rm max}={\rm col}(\lambda_1^{\rm max},\cdots,\lambda_n^{\rm max})$, $\lambda_{\rm min}={\rm col}(\lambda_1^{\rm min},\cdots,\lambda_n^{\rm min})$, $\zeta_{\rm max}={\rm col}(\zeta_1^{\rm max},\cdots,\zeta_n^{\rm max})$ and $\zeta_{\rm min}={\rm col}(\zeta_1^{\rm min},\cdots,\zeta_n^{\rm min})$, we now for \eqref{eq:Model-based_Program} define the Lagrangian 
\begin{IEEEeqnarray}{rCl}
\mathbf{L}(\mathcal{X}) &=&g(u)+\zeta_{\rm max}^\top   (G_{\rm red} u+w - I_{\rm dc}^{\rm max})\IEEEnonumber\\
&& +\: \zeta_{\rm min}^\top   (I_{\rm dc}^{\rm min}-G_{\rm red} u - w) \IEEEnonumber\\
&& +\: \lambda_{\rm max}^\top (V_{\rm nom}1_n + u - V_{{\rm dc}}^{\rm max}) \IEEEnonumber\\
&& +\: \lambda_{\rm min}^\top (V_{{\rm dc}}^{\rm min}-V_{\rm nom}1_n - u ), \IEEEyesnumber\label{eq:Lagr}
\end{IEEEeqnarray}
where $\mathcal{X}={\rm col}(u,\zeta_{\rm max},\zeta_{\rm min},\lambda_{\rm max},\lambda_{\rm min})$. Assume that the Slater's conditions are satisfied for the problem in hand, e.g., it is a quadratic minimization problem with linear constraints. Then, the KKT conditions are necessary and sufficient for optimality. In other words, the above variables are optimal, if and only if they satisfy the KKT conditions\cite[Ch. 5.5]{Boyd}
\begin{IEEEeqnarray}{rCll}
\label{eq:KKTConditions}
0_n&=& \begin{bmatrix}
\mathcal{I}_n \!&\! G_{\rm red}^\top
\end{bmatrix} \nabla f(u,y) \IEEEnonumber\\
&&+\: G_{\rm red}^\top (\zeta_{\rm max} - \zeta_{\rm min}) + (\lambda_{\rm max}-\lambda_{\rm min}),&\IEEEyesnumber\IEEEyessubnumber\label{eq:KKTStationary}\\
0&=&\lambda_i^{\rm max}(V_{\rm nom} + u_i -V_{{\rm dc}i}^{\rm max}),\quad &\forall i \in \mathcal{N},\qquad\IEEEyessubnumber\label{eq:KKTLambdaMax}\\
0&=&\lambda_i^{\rm min}(V_{{\rm dc}i}^{\rm min}-V_{\rm nom} - u_i ),\quad &\forall i \in \mathcal{N},\quad\IEEEyessubnumber\label{eq:KKTLambdaMin}\\
0&=&\zeta_i^{\rm max}(y_i-I_{{\rm dc}i}^{\rm max}),\quad &\forall i \in \mathcal{N},\quad\IEEEyessubnumber\label{eq:KKTZetaMax}\\
0&=&\zeta_i^{\rm min}(I_{{\rm dc}i}^{\rm min}-y_i),\quad &\forall i \in \mathcal{N},\quad\IEEEyessubnumber\label{eq:KKTZetaMin}\\
0&\leq& \lambda_i^{\rm max},\,\, 0\leq \lambda_i^{\rm min},\,\, 0\leq \zeta_i^{\rm max},\,\, 0\leq \zeta_i^{\rm min},\quad &\forall i \in \mathcal{N},\quad\IEEEyessubnumber\label{eq:KKTDualPositive}\\
&&V_{{\rm dc}i}^{\rm min}\leq V_{\rm nom} + u_i \leq V_{{\rm dc}i}^{\rm max},\quad &\forall i \in \mathcal{N},\quad\IEEEyessubnumber\label{eq:KKTVoltageConstraints}\\
&&I_{{\rm dc}i}^{\rm min}\leq y_i \leq I_{{\rm dc}i}^{\rm max},\quad &\forall i \in \mathcal{N},\quad\IEEEyessubnumber\label{eq:KKTCurrentConstraints}
\end{IEEEeqnarray}
where $\nabla f(u,y)$ is the gradient of $f(u,y)$ and $0_n=\mathrm{col}(0,\ldots,0)$. In the above equations we used $y=G_{\rm red} u+w$ and $\nabla g(u)=\begin{bmatrix}
\mathcal{I}_n \!&\! G_{\rm red}^\top
\end{bmatrix} \nabla f(u,y)$. In the next section, we propose a controller based on primal-dual dynamics \cite{Cherukuri2016} such that this set of optimality conditions is an attractive steady state of the closed-loop system.
\begin{remark}
In \eqref{eq:Feedback-based_Program}, we considered the specific inequality constraints of input-output limits. However, one may consider the general constraint of $A_{\rm const} u \leq b_{\rm const}$ without having to make significant changes in our upcoming formulations. An example of such a consideration is line current limits.
\end{remark}
\section{Proposed Central Secondary Control}
\label{Sec:ProposedController}
We are now ready to introduce our secondary controller. We first present the controller which is based on the primal-dual dynamics. Then, we study the stability of the closed-loop system and show that under the proposed controller, the system asymptotically converges to a steady state where the KKT conditions \eqref{eq:KKTConditions} are satisfied. We also offer a communication-triggering mechanism to reduce communication traffic between the central secondary controller the MMC stations.

\subsection{Optimal Steady-State Control with Primal-Dual Dynamics}
\begin{figure}
\centering
\includegraphics[width=0.9\columnwidth]{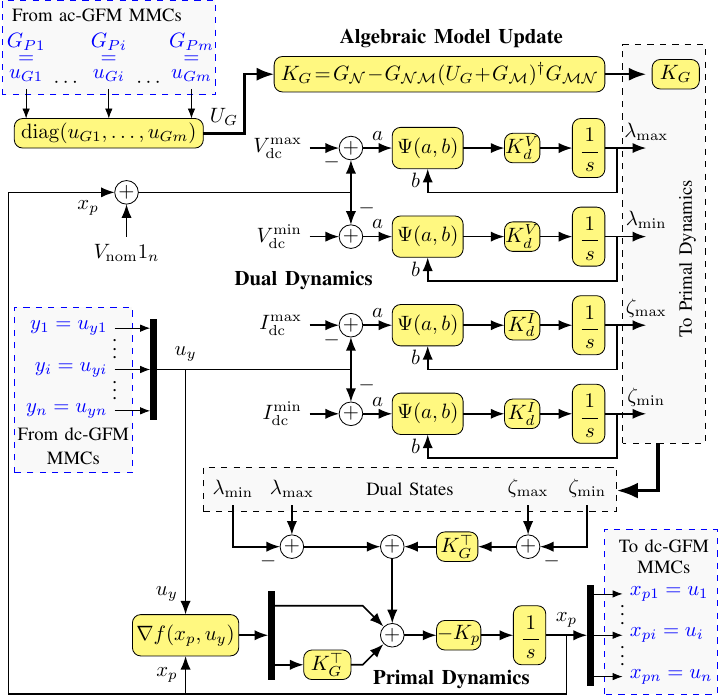}
\caption{Schematic diagram of the proposed secondary controller based on primal-dual dynamics in \eqref{eq:PDD}, in vector form.\label{Fig:SecondaryControl}}
\end{figure}
Given the characterization of the optimal solutions in \eqref{eq:KKTConditions}, it is natural to use the primal-dual dynamics to find them \cite{Cherukuri2016}. Accordingly, we propose the secondary controller
\begin{IEEEeqnarray}{rCl}
\label{eq:PDD}
u &=& x_p,\IEEEyesnumber\IEEEyessubnumber\label{eq:SecondaryInput}\\
\dot{x}_p&=& -K_p\begin{bmatrix}
\mathcal{I}_n \!&\! K_G^\top
\end{bmatrix} \nabla f(x_p,u_y)  \IEEEnonumber\\
&&- K_p K_G^\top (\zeta_{\rm max} - \zeta_{\rm min})-K_p(\lambda_{\rm max}-\lambda_{\rm min}),\quad\IEEEyessubnumber\label{eq:PrimalDynamics}\\
\dot{\zeta}_{\rm max}&=&K_d^I \Psi \big( u_y-I_{\rm dc}^{\rm max},\zeta_{\rm max} \big),\IEEEyessubnumber\label{eq:DualZetaMax}\\
\dot{\zeta}_{\rm min}&=& K_d^I \Psi \big( I_{\rm dc}^{\rm min}-u_y,\zeta_{\rm min} \big),\IEEEyessubnumber\label{eq:DualZetaMin}\\
\dot{\lambda}_{\rm max}&=&K_d^V \Psi \big( V_{\rm nom}1_n + x_p - V_{{\rm dc}}^{\rm max},\lambda_{\rm max} \big),\quad\,\,\,\IEEEyessubnumber\label{eq:DualLambdaMax}\\
\dot{\lambda}_{\rm min}&=&K_d^V \Psi \big(V_{{\rm dc}}^{\rm min}-V_{\rm nom}1_n - x_p ,\lambda_{\rm min} \big),\IEEEyessubnumber\label{eq:DualLambdaMin}\\
u_y &=& y,\IEEEyessubnumber\label{eq:PrimalInput}
\end{IEEEeqnarray}
where the diagonal matrices $K_p,K_d^I,K_d^V\succ 0$ are tunable and for any $a={\rm col}(a_1,\cdots,a_n)$ and $b={\rm col}(b_1,\cdots,b_n)$, the vector function $\Psi:\mathbb{R}^{2n}\mapsto \mathbb{R}^n$ is defined as
\begin{IEEEeqnarray}{C}
\Psi(a,b)={\rm col}\Big( \psi(a_1,b_1),\ldots,\psi(a_n,b_n) \Big),\IEEEyesnumber\IEEEyessubnumber\label{eq:Psi}\\
\text{where }\psi(a_i,b_i)=\begin{cases}
a_i & \text{if } b_i>0\\
\max\{0,a_i\} & \text{if } b_i\leq 0
\end{cases}\quad\forall i\in\mathcal{N}.\qquad\IEEEyessubnumber\label{eq:psi}
\end{IEEEeqnarray}
The matrix $K_G$ is an online construction of $G_{\rm red}$ which is algebraically obtained by using the Moore-Penrose inverse as
\begin{IEEEeqnarray}{rCl}
\label{eq:AlgebraicConstruction}
K_G &=& G_\mathcal{N} - G_\mathcal{NM} (G_\mathcal{M}+U_G)^\dagger G_\mathcal{MN},\IEEEyesnumber\IEEEyessubnumber\label{eq:ModelMatrix}\\
U_G &=& G_P.\IEEEyessubnumber\label{eq:ConductanceInput}
\end{IEEEeqnarray}
where $U_G={\rm diag}(u_{G1},\ldots,u_{Gm})={\rm diag}(G_{P1},\ldots,G_{Pm})\in \mathbb{R}^{m \times m}$ is a diagonal input matrix of the conductances in \eqref{eq:ACGFM-QS-Linear} received from the ac-GFM stations.
Setting $\dot{x}_p=\dot{\zeta}_{\rm max}=\dot{\zeta}_{\rm min}=\dot{\lambda}_{\rm max}=\dot{\lambda}_{\rm min}=0_n$, one can establish that under this controller, the system achieves a steady state where the KKT conditions \eqref{eq:KKTConditions} are satisfied; i.e., the controller acts as an online feedback optimizer that steers the operation of the system towards optimality in real-time.
Fig.~\ref{Fig:SecondaryControl} shows a schematic diagram of the controller which exchanges information with the underlying subsystems shown in Fig.~\ref{Fig:ac-GFM} and Fig.~\ref{Fig:dc-GFM}.

\subsection{Convergence and Stability Analysis}
\label{Sec:Stability}

The closed-loop system consists of \eqref{eq:Input-OutputMap} and \eqref{eq:PDD} and we analyze the stability of its steady state. We first make the system stability statement and then use the convergence analysis of primal-dual dynamics in\cite[Sec. 4]{Cherukuri2016} to prove it. The following theorem states the stability for the closed-loop system.
\begin{theorem}
If the problem \eqref{eq:Model-based_Program} is convex, then under the primal-dual dynamics \eqref{eq:PDD}-\eqref{eq:AlgebraicConstruction}, the system \eqref{eq:Input-OutputMap} asymptotically converges to a steady state that is the optimal solution of \eqref{eq:Model-based_Program}.
\end{theorem}
\begin{proof}
Considering $x_d={\rm blkcol}(\zeta_{\rm max},\zeta_{\rm min},\lambda_{\rm max},\lambda_{\rm min})$, let us define the new Lagrangian function
\begin{IEEEeqnarray}{rCl}
\mathbf{L}_{\rm new}(x_p,x_d) &=& -\mathbf{L}(x_p,\zeta_{\rm max},\zeta_{\rm min},\lambda_{\rm max},\lambda_{\rm min}),\IEEEyesnumber\label{eq:NewLagrangian}
\end{IEEEeqnarray}
where $\mathbf{L}(\cdot)$ is given in \eqref{eq:Lagr}. The Lagrangian $\mathbf{L}(\mathcal{X})$ in \eqref{eq:Lagr} is \textit{convex} in $x_p$ and \textit{linear} in the dual variables $x_d={\rm blkcol}(\zeta_{\rm max},\zeta_{\rm min},\lambda_{\rm max},\lambda_{\rm min})$. Thus, the new Lagrangian $\mathbf{L}_{\rm new}$ is \textit{concave} in $x_p$ and \textit{convex} in $x_d$. We can also write \eqref{eq:PDD}-\eqref{eq:AlgebraicConstruction} in the stacked form
\begin{IEEEeqnarray}{rCl}
\label{eq:PDDStacked}
\dot{x}_p &=& K_p\nabla_{x_p} \mathbf{L}_{\rm new}(x_p,x_d),\IEEEyesnumber\IEEEyessubnumber\label{eq:PrimalStacked}\\
\dot{x}_d &=& K_d \Pi \big(-\nabla_{x_d}\mathbf{L}_{\rm new}(x_p,x_d),x_d\big),\IEEEyessubnumber\label{eq:DualStacked}
\end{IEEEeqnarray}
where $K_d={\rm blkdiag}(K_d^I,K_d^I,K_d^V,K_d^V)$ and where for any $a={\rm col}(a_1,\cdots,a_{4n})$ and $b={\rm col}(b_1,\cdots,b_{4n})$ the function $\Pi:\mathbb{R}^{8n}\mapsto \mathbb{R}^{4n}$ is defined as
\begin{IEEEeqnarray}{rCl}
\Pi(a,b)&=& {\rm col} \Big ( \psi(a_1,b_1),\ldots,\psi(a_{4n},b_{4n}) \Big ) ,\IEEEyesnumber\label{eq:Pi}
\end{IEEEeqnarray}
where $\psi(a_i,b_i)$ is defined in \eqref{eq:psi}. The Lagrangian \eqref{eq:NewLagrangian} and the primal-dual dynamics \eqref{eq:PDDStacked} are consistent with \cite[Eqs. (3)-(6)]{Cherukuri2016}. Thus, with $x_p^*$ and $x_d^*$ being the optimal values and considering the Lyapunov function
\begin{IEEEeqnarray}{rCl}
\mathcal{V}(x_p,x_d) &=& \tfrac{1}{2} (x_p-x_p^*)^\top K_p^{-1} (x_p-x_p^*)\IEEEnonumber\\
&&\qquad\qquad + \tfrac{1}{2} (x_d-x_d^*)^\top K_d^{-1} (x_d-x_d^*),\IEEEnonumber
\end{IEEEeqnarray}
Theorem 4.5 in\cite{Cherukuri2016} can be applied and its results can be concluded; i.e., under the primal-dual dynamics \eqref{eq:PDD}-\eqref{eq:AlgebraicConstruction}, the set of primal-dual solutions of the problem \eqref{eq:Model-based_Program} is asymptotically stable, and convergence of each solution is to a point. Thus, the steady state of the closed-loop system, which is the optimal solution to this problem, is asymptotically stable.
\end{proof}

\subsection{Reduction of Communication Traffic}
The controller \eqref{eq:PDD}-\eqref{eq:AlgebraicConstruction}, continuously receives $y_i$ from every dc-GFM station $i\in \mathcal{N}$ and $G_{Pi}$ from every ac-GFM station $i\in \mathcal{M}$ and it forwards the elements of the vector $x_p$ to the dc-GFM MMCs continuously. Implementing this feedback system requires high-frequency sampled communications between the central secondary controller and the MMCs, which may be redundant and not necessary. This is particularly true in steady state where the signals do not change significantly. In this subsection, we propose a mechanism to reduce communication traffic and avoid redundant data transmissions.

In place of the continuously-varying signals $y$, $x_p$, and $G_P$ in \eqref{eq:SecondaryInput}, \eqref{eq:PrimalInput}, and \eqref{eq:ConductanceInput}, we use the piece-wise constant signals $\hat{y}={\rm col}(\hat{y}_1,\ldots,\hat{y}_n)$, $\hat{x}_p={\rm col}(\hat{x}_1^p,\ldots,\hat{x}_n^p)$, and $\hat{G}_P={\rm diag}(\hat{G}_{P1},\ldots,\hat{G}_{Pn})$ and modify these inputs as
\begin{IEEEeqnarray}{c}
u_y=\hat{y}\quad\text{ and }\quad u = \hat{x}_p\quad\text{ and }\quad U_G = \hat{G}_P,\label{eq:ReducedInputs}
\end{IEEEeqnarray}
where the elements of $\hat{y}$, $\hat{x}_p$, and $\hat{G}_P$ are some piece-wise constant signals in time $t$ that are composed of the latest sampled values of $y_i(t)$, $x_i^p(t)$, and $G_{Pi}(t)$ as defined by
\begin{IEEEeqnarray}{rCll}
\label{eq:PWCSignals}
\hat{y}_i(t) &=& y_i(t_{l}^{y_i}),&\quad\forall t\in[t_{l}^{y_i},t_{l+1}^{y_i}),\IEEEyesnumber\IEEEyessubnumber\label{eq:PWCI}\\
\hat{x}_i^p(t) &=& x_i^p(t_{c}^{x_i}),&\quad\forall t\in[t_{c}^{x_i},t_{c+1}^{x_i}),\IEEEyessubnumber\label{eq:PWCxp}\\
\hat{G}_{Pi}(t) &=& G_{Pi}(t_{k}^{G_i}),&\quad\forall t\in[t_{k}^{G_i},t_{k+1}^{G_i}),\IEEEyessubnumber\label{eq:PWCGp}
\end{IEEEeqnarray}
and $t_{0}^{y_i}$, $t_{0}^{x_i}$, and $t_{0}^{G_i}$ are the first sampling (communication) instants at the activation time of the controller. Equation \eqref{eq:PWCI} shows that the communicated signal $\hat{y}_i$ is held constant between every consecutive sampling instants $t_{l}^{y_i}$ and $t_{l+1}^{y_i}$. Similarly, \eqref{eq:PWCxp} (resp. \eqref{eq:PWCGp}) shows that the signal $\hat{x}_i^p$ (resp. $\hat{G}_{Pi}$) is held constant between every consecutive sampling instants $t_{c}^{x_i}$ and $t_{c+1}^{x_i}$ (resp. $t_{k}^{G_i}$ and $t_{k+1}^{G_i}$). Let ${\rm inf}(\cdot)$ denote the infimum and $\vee$ the logical or. The next sampling instants $t_{l+1}^{y_i}$, $t_{c+1}^{x_i}$, and $t_{k+1}^{G_i}$ are then determined by
\begin{IEEEeqnarray}{rCl}
\label{eq:SamplingInstants}
t_{l+1}^{y_i} &=& {\rm inf}\{ t > t_{l}^{y_i}+T_{\rm min} \:|\: \mathcal{C}_\sigma^y \vee \mathcal{C}_T^y\},\qquad\IEEEyesnumber\IEEEyessubnumber\label{eq:SamplingInstanty}\\
t_{c+1}^{x_i} &=& {\rm inf}\{ t > t_{c}^{x_i}+T_{\rm min} \:|\: \mathcal{C}_\sigma^x \vee \mathcal{C}_T^x\},\qquad\IEEEyessubnumber\label{eq:SamplingInstantxp}\\
t_{k+1}^{G_{i}} &=& {\rm inf}\{ t > t_{k}^{G_{i}}+T_{\rm min} \:|\: \mathcal{C}_\sigma^G \vee \mathcal{C}_T^G\},\qquad\IEEEyessubnumber\label{eq:SamplingInstantG}
\end{IEEEeqnarray}
where $\mathcal{C}_\sigma^y$, $\mathcal{C}_T^y$, $\mathcal{C}_\sigma^x$, $\mathcal{C}_T^x$, $\mathcal{C}_\sigma^G$, $\mathcal{C}_T^G$ are the triggering conditions
\begin{IEEEeqnarray}{lCl}
\mathcal{C}_\sigma^y: |y_i(t)\!-\!y_i(t_{l}^{y_i})|>\sigma_y, &\quad\qquad& \mathcal{C}_T^y: t-t_{l}^{y_i}>T_{\rm max},\IEEEnonumber\\
\mathcal{C}_\sigma^x: |x_i^p(t)\!-\!x_i^p(t_{c}^{x_i})|>\sigma_x, &\quad\qquad& \mathcal{C}_T^x: t-t_{c}^{x_i}>T_{\rm max},\IEEEnonumber\\
\mathcal{C}_\sigma^G: |G_{Pi}(t)\!-\!G_{Pi}(t_{k}^{G_{i}})|>\sigma_G, &\quad\qquad& \mathcal{C}_T^G: t-t_{k}^{G_{i}}>T_{\rm max}.\IEEEnonumber
\end{IEEEeqnarray}

Equation \eqref{eq:SamplingInstanty} means that, for each time interval, the next sampling instant $t_{l+1}^{y_i}$ is at least $T_{\rm min}$ and at most $T_{\rm max}$ seconds later than the previous instant $t_{l}^{y_i}$, and it may be when the signal $y_i(t)$ deviates so much from its latest sampled value $y_i(t_{l}^{I_i})$ that the error $|y_i(t)-y_i(t_{l}^{y_i})|$ violates the threshold $\sigma_y$. In fact, the inter-sampling time intervals take a value, depending on the error $|y_i(t)-y_i(t_{l}^{y_i})|$, between the lower bound $T_{\rm min}$ and the upper bound $T_{\rm max}$, i.e., $t_{l+1}^{y_i}-t_{l}^{y_i}\in[T_{\rm min}, T_{\rm max}]$. A similar argument can be made for equations \eqref{eq:SamplingInstantxp}-\eqref{eq:SamplingInstantG} and therefore $t_{c+1}^{x_i}-t_{c}^{x_i}\in[T_{\rm min},T_{\rm max}]$ and $t_{k+1}^{G_i}-t_{k}^{G_i}\in[T_{\rm min},T_{\rm max}]$. The entire scheme is shown in Fig~\ref{Fig:CommScheme} while Algorithms~\ref{LocalDCAlgorithm}-\ref{CentralAlgorithm} describe the proposed central and local sampling mechanisms. It should be noted that, in these algorithms, we also use the following piece-wise constant time signals as decision variables.
\begin{IEEEeqnarray}{rCl}
\label{eq:PWCtimeSignals}
\hat{t}_i^y(t) &=& t_{l}^{y_i},\quad\forall t\in[t_{l}^{y_i},t_{l+1}^{y_i}),\IEEEyesnumber\IEEEyessubnumber\label{eq:PWCtimeI}\\
\hat{t}_i^{x}(t) &=& t_{c}^{x_i},\quad\forall t\in[t_{c}^{x_i},t_{c+1}^{x_i}),\IEEEyessubnumber\label{eq:PWCtimexp}\\
\hat{t}_i^{G}(t) &=& t_{k}^{G_i},\quad\forall t\in[t_{k}^{G_i},t_{k+1}^{G_i}).\IEEEyessubnumber\label{eq:PWCtimeG}
\end{IEEEeqnarray}

\begin{figure}
\centering
\includegraphics[width=0.9\columnwidth]{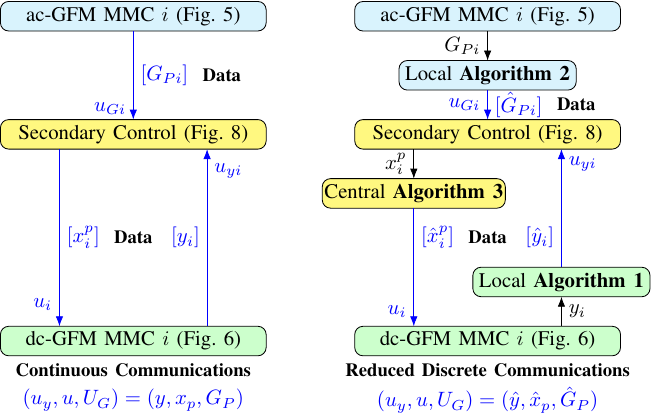}
\caption{Continuous vs. reduced discrete data communications based on the see Algorithms~\ref{LocalDCAlgorithm}-\ref{CentralAlgorithm} for the sampled communication details.\label{Fig:CommScheme}}
\end{figure}
\begin{algorithm}
\small
\caption{Sampling at the $i$th dc-GFM station.}\label{LocalDCAlgorithm}
\begin{algorithmic}[1]
\Statex\textbf{initialization:}
\State Define the memory states $\hat{y}_i$ and $\hat{t}_i^y$
\Statex \Comment{The first sampling instant}
\State Sample: $\hat{y}_i\gets y_i$ and $\hat{t}_{i}^y\gets t$
\State Store the samples $\hat{y}_i$ and $\hat{t}_{i}^y$ and hold them constant
\State Send $\hat{y}_i$ to the central supervisory controller
\While{supervisory control is active}
\State Compute $e_i^y=\hat{y}_i-y_i$ and $\Delta t_i^{y}= t - \hat{t}_{i}^y$
\If{\{$\Delta t_i^{y} \geq T_{\rm min}$\} \textbf{and} \{$|e_i^y|>\sigma_y  \:\textbf{or}\: \Delta t_i^{y} >T_{\rm max}$\} \text{\color{white}. .....} }\text{ }\Comment{Sampling instant}
\State Do lines 2 to 4 of the algorithm
\EndIf
\EndWhile
\end{algorithmic}
\end{algorithm}
\begin{algorithm}
\small
\caption{Sampling at the $i$th ac-GFM station.}\label{LocalACAlgorithm}
\begin{algorithmic}[1]
\Statex\textbf{initialization:}
\State Define the memory states $\hat{G}_{Pi}$ and $\hat{t}_i^G$
\Statex \Comment{The first sampling instant}
\State Sample: $\hat{G}_{Pi}\gets G_{Pi}$ and $\hat{t}_{i}^G\gets t$
\State Store the samples $\hat{G}_{Pi}$ and $\hat{t}_{i}^G$ and hold them constant
\State Send $\hat{G}_{Pi}$ to the central supervisory controller
\While{supervisory control is active}
\State Compute $e_i^G=\hat{G}_{Pi}-G_{Pi}$ and $\Delta t_i^G= t - \hat{t}_{i}^G$
\If{\{$\Delta t_i^G \geq T_{\rm min}$\} \textbf{and} \{$|e_i^G|>\sigma_G  \:\textbf{or}\: \Delta t_i^G >T_{\rm max}$\} \text{\color{white}. .....} }\text{ }\Comment{Sampling instant}
\State Do lines 2 to 4 of the algorithm
\EndIf
\EndWhile
\end{algorithmic}
\end{algorithm}
\begin{algorithm}
\small
\caption{The central sampling.}\label{CentralAlgorithm}
\begin{algorithmic}[1]
\Statex\textbf{initialization:}
\State Define the memory states $\hat{x}_i^p$ and $\hat{t}_i^{x}$ for all $i=1,\cdots,n$
\ForAll{$i$}\Comment{The first sampling instant}
\State Sample: $\hat{x}_i^p\gets x_i^p$ and $\hat{t}_i^{x}\gets t$
\State Store the samples $\hat{x}_i^p$ and $\hat{t}_{i}^{x}$ and hold them constant
\State Send $\hat{x}_i^p$ to the $i$th station
\EndFor
\While{supervisory control is active \textbf{for all} $i$}
\State Compute $e_i^{x}=\hat{x}_i^p-x_i^p$ and $\Delta t_i^{x}= t - \hat{t}_{i}^{x}$
\If{\{$\Delta t_i^{x} \geq T_{\rm min}$\} \textbf{and} \{$|e_i^{x}|>\sigma_{x} \:\textbf{or}\: \Delta t_i^{x} >T_{\rm max}$\} \text{ }\qquad\quad}\Comment{Sampling instant}
\State Do lines 3 to 5 of the algorithm
\EndIf
\\\textbf{end for}
\EndWhile
\end{algorithmic}
\end{algorithm}
\section{Case Studies}
\label{Sec:CaseStudies}
To verify the effectiveness of the proposed controller, we apply it to the offshore MT-HVdc grid shown in Fig.~\ref{Fig:HVdcGrid}, simulated in MATLAB/Simulink environment. This grid is composed of 12 MMC stations, modeled according to Fig.~\ref{Fig:MMC}, which are interconnected via the lines modeled by Fig.~\ref{Fig:Cable}. The stations 1 to 6 (color coded by green) are connected to the areas 1 to 6 and operate in the dc-GFM mode (Fig.~\ref{Fig:dc-GFM}) while the stations 7 to 12 (color coded by blue) integrate the offshore energy hubs into the HVdc grid and operate in ac-GFM mode (Fig.~\ref{Fig:ac-GFM}). The electrical and control parameters of different components are summarized in Table~\ref{Table:Electric} and Table~\ref{Table:Control}. It is to be noted that the areas 1 to 6, represented by the ac-side box in Fig.~\ref{Fig:MMC}, are modeled by ideal 3-phase voltage sources connected to the ac capacitors via RL sections with the parameters $(L_g,R_g)$ given in the same table. For each offshore hub station, the ac side in Fig.~\ref{Fig:MMC} is modeled by an ideal power source with user-defined active and reactive power setpoints given in Fig.~\ref{Fig:PowerRef}.

\begin{table}
\centering
\caption{Electrical Parameters of the System}
\label{Table:Electric}
\begin{tabular}{|c|c|c|c|c|c|c|}
\hline
\multicolumn{7}{|c|}{\cellcolor{green!20} Area dc-GFM MMC Stations} \\ \hline
\cellcolor{green!20} Station \#           &\cellcolor{green!20} 1 &\cellcolor{green!20} 2 &\cellcolor{green!20} 3 &\cellcolor{green!20} 4 &\cellcolor{green!20} 5 &\cellcolor{green!20} 6 \\ \hline
$S_\star\,{\rm[MVA]}$  & 2000 & 1500 & 750 & 500 & 300 & 1000\\ \hline
$C_{\rm ac}\,{\rm[\mu F]}$  & 5.76 & 4.32 & 2.16 & 1.44 & 0.864 & 2.88\\ \hline
$L_f\,{\rm[H]}$  & 0.09 & 0.08 & 0.055 & 0.04 & 0.06 & 0.05\\ \hline
$R_f\,{\rm[\Omega]}$  & 0.3 & 0.25 & 0.3 & 0.27 & 0.25 & 0.28\\ \hline
$C_{\rm eq}\,{\rm[\mu F]}$  & 57.808 & 43.356 & 21.678 & 14.452 & 8.671 & 28.904\\ \hline
$L_{\rm arm}\,{\rm[H]}$  & 0.029 & 0.032 & 0.035 & 0.02 & 0.03 & 0.025\\ \hline
$R_{\rm arm}\,{\rm[\Omega]}$  & 0.53 & 0.49 & 0.57 & 0.45 & 0.55 & 0.5\\ \hline
$L_g\,{\rm[H]}$  & 0.045 & 0.032 & 0.041 & 0.03 & 0.04 & 0.035\\ \hline
$R_g\,{\rm[\Omega]}$  & 0.015 & 0.01 & 0.013 & 0.009 & 0.012 & 0.01\\ \hline
\end{tabular}
\\\vspace*{1mm}
\begin{tabular}{|c|c|c|c|c|c|c|}
\hline
\multicolumn{7}{|c|}{\cellcolor{cyan!15} Offshore ac-GFM MMC Stations} \\ \hline
\cellcolor{cyan!15}Station \#&\cellcolor{cyan!15} 7 &\cellcolor{cyan!15} 8 &\cellcolor{cyan!15} 9 &\cellcolor{cyan!15} 10 &\cellcolor{cyan!15} 11 &\cellcolor{cyan!15} 12 \\ \hline
$S_\star\,{\rm[MVA]}$  & 1000 & 1000 & 1000 & 1000 & 1000 & 1000\\ \hline
$C_{\rm ac}\,{\rm[\mu F]}$  & 2.88 & 3.168 & 2.736 & 3.024 & 2.592 & 2.8224\\ \hline
$L_f\,{\rm[H]}$  & 0.058 & 0.052 & 0.053 & 0.047 & 0.055 & 0.05\\ \hline
$R_f\,{\rm[\Omega]}$  & 0.39; & 0.29 & 0.25 & 0.31 & 0.33 & 0.32\\ \hline
$C_{\rm eq}\,{\rm[\mu F]}$  & 31.794 & 28.904 & 28.037 & 30.349 & 28.62 & 27.748\\ \hline
$L_{\rm arm}\,{\rm[H]}$  & 0.035 & 0.034 & 0.025 & 0.029 & 0.028 & 0.03\\ \hline
$R_{\rm arm}\,{\rm[\Omega]}$  & 0.55 & 0.52 & 0.49 & 0.46 & 0.49 & 0.51\\ \hline
\end{tabular}
\\\vspace*{1mm}
\begin{tabular}{|c|c|c|c|}
\hline
\multicolumn{4}{|c|}{\cellcolor{gray!20} Transmission Line Model Parameters (Fig.~\ref{Fig:Cable})} \\ \hline
\cellcolor{gray!20}$r_1\,{\rm[\Omega/km]}$ &\cellcolor{gray!20} $r_2\,{\rm[\Omega/km]}$ &\cellcolor{gray!20} $r_3\,{\rm[\Omega/km]}$ &\cellcolor{gray!20} $g\,{\rm[\mu \Omega^{-1}/km]}$ \\ \hline
0.1265 & 0.1504 & 0.0178 & 0.1015 \\ \hline
\cellcolor{gray!20}$l_1\,{\rm[m H/km]}$ &\cellcolor{gray!20} $l_2\,{\rm[m H/km]}$ & \cellcolor{gray!20} $l_3\,{\rm[m H/km]}$ &\cellcolor{gray!20}  $c\,{\rm[\mu F/km]}$ \\ \hline
0.2644 & 7.2865 & 3.6198 & 0.1616 \\ \hline
\end{tabular}
\end{table}
\begin{table}
\centering
\caption{Control Parameters of the System}
\label{Table:Control}
\begin{tabular}{c}
\hline
\cellcolor{gray!20} MMC Station Control Parameters \\ \hline\\[-0.8em]
$V_{\rm dc}^{\rm nom}=620\,{\rm [kV]},\, V_{\rm ac}^{\rm nom}= 313.5\,{\rm[kV]\, RMS},\,f_{\rm nom}=50\,{\rm [Hz]} $\vspace{0.7mm} \\\hline\\[-0.8em]
$ ( k_{PC}^{dq}, k_{IC}^{dq}, k_{PC}^{dc}, k_{IC}^{dc} ) = 1000 ( L_{\rm eq}, R_{\rm eq}, L_{\rm arm}, R_{\rm arm} )$\vspace{0.7mm} \\\hline\\[-0.8em]
$ k_{IV}^{dq} = k_{IV}^{dc} = 5,\,k_{PV}^{dq}=2\sqrt{10 C_{\rm ac}},\, k_{PV}^{dc}=\sqrt{25 C_{\rm dc}}$\vspace{0.7mm} \\\hline\\[-0.8em]
$ k_{IW} = 500,\,  k_{PW}=2\sqrt{k_{IW}},\, W_{\Sigma z}^{\rm set} = C_{\rm eq}(1.25V_{\rm dc}^{\rm nom} )^2$\vspace{0.7mm} \\\hline\\[-0.8em]
$Q_{\rm ac}^{\rm set}=P_{\rm ac}^{\rm set}=0,\,m_Q \!=\! 0.05V_{\rm ac}^{\rm nom}/S_\star,\, m_P \!=\! 0.005 f_{\rm nom}/S_\star\!\!$\vspace{0.7mm} \\\hline\\[-0.8em]
$(\tau_V,\tau_\omega,\tau_G) = (1,0.1,0.03),\,k_{IQ}=2,\,k_{PQ}=0.005k_{IQ}$\vspace{0.7mm} \\\hline\\[-0.8em]
$I_{\star}^{\rm dc} = S_\star/V_{\rm dc}^{\rm nom},\,R_d \!=\! 0.05 V_{\rm dc}^{\rm nom}/I_{\rm dc}^{\rm rated}$\vspace{0.7mm} \\\hline
\end{tabular}
\\\vspace*{1mm}
\begin{tabular}{l|l}
\hline
\multicolumn{2}{c}{\cellcolor{yellow!50} Cost Functions and Primal-Dual Gains for Different Objectives}\\ \hline\\[-0.8em]
$\begin{matrix}
   \qquad\text{Loss Reduction:}
\end{matrix}$ & $\begin{matrix}
    f(u,y)=\tfrac{1}{2}u^\top P_u  u, \quad   P_u= K_G-K_G 1_6\\
    K_p=200\mathcal{I}_6,\,K_d^I=25\mathcal{I}_6,\,K_d^V=25\mathcal{I}_6
\end{matrix}$\\[-0.8em] \\\hline\\[-0.8em]
$\begin{matrix}
   \text{General Quadratic}\\
   \text{Output Programming:}
\end{matrix}$ & $\begin{matrix}
   {\color{white}\Big(} f(u,y)=\tfrac{1}{2}y^\top P_y  y + q_y^\top   y \\
                        P_y={\rm diag}(2.4,5.7,3,4.2,3.6,4.8)\\
   {\color{white}\Big(} q_y=1000\times  {\rm col}(30, 75, 36, 54, 45, 63)\\
   K_p=200\mathcal{I}_6,\,K_d^I=10\mathcal{I}_6,\,K_d^V=10\mathcal{I}_6
\end{matrix}$\\[-0.8em] \\\hline\\[-0.8em]
$\begin{matrix}
   \text{Proportional}\\
   \text{Current Minimization:}
\end{matrix}$ & $\begin{matrix}
    f(u,y)=\tfrac{1000 }{2}y^\top P_y  y\quad\\
    {\color{white}\Big(} P_y={\rm diag}(\tfrac{1}{I_{\star 1}^{\rm dc}},\tfrac{1}{I_{\star 2}^{\rm dc}},\tfrac{1}{I_{\star 3}^{\rm dc}},\tfrac{1}{I_{\star 4}^{\rm dc}},\tfrac{1}{I_{\star 5}^{\rm dc}},\tfrac{1}{I_{\star 6}^{\rm dc}})\\
    K_p=200\mathcal{I}_6,\,K_d^I=10\mathcal{I}_6,\,K_d^V=10\mathcal{I}_6
\end{matrix}$\\[-0.8em] \\\hline
\multicolumn{2}{c}{\cellcolor{yellow!50} Triggering Mechanism Parameters} \\ \hline\\[-0.8em]
Thresholds: & $ ( \sigma_y , \sigma_x , \sigma_G ) = ( 5, 20, 0.0001 )$\\[-0.8em] \\\hline\\[-0.8em]
Time Thresholds: & $ (T_{\rm min} , T_{\rm max} ) = ( 0.01{\rm s}, 1{\rm s}) $\\[-0.8em] \\\hline
\end{tabular}
\end{table}

\begin{figure}
\centering
\includegraphics[width=0.493\columnwidth]{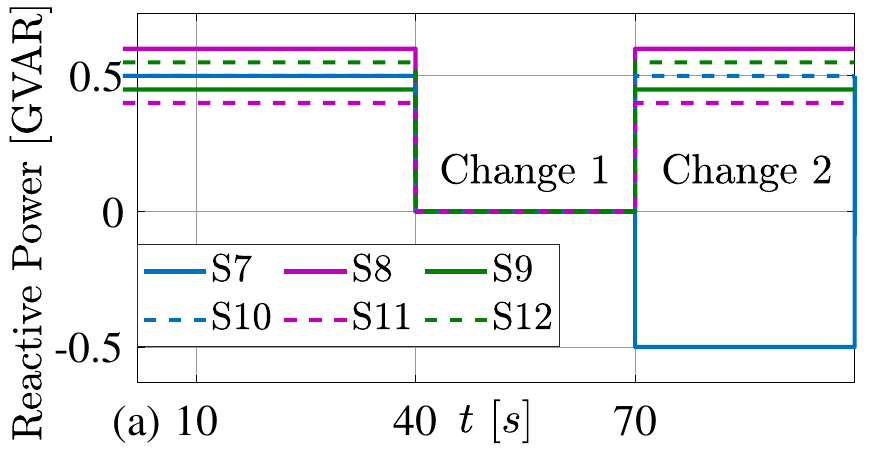}
\includegraphics[width=0.493\columnwidth]{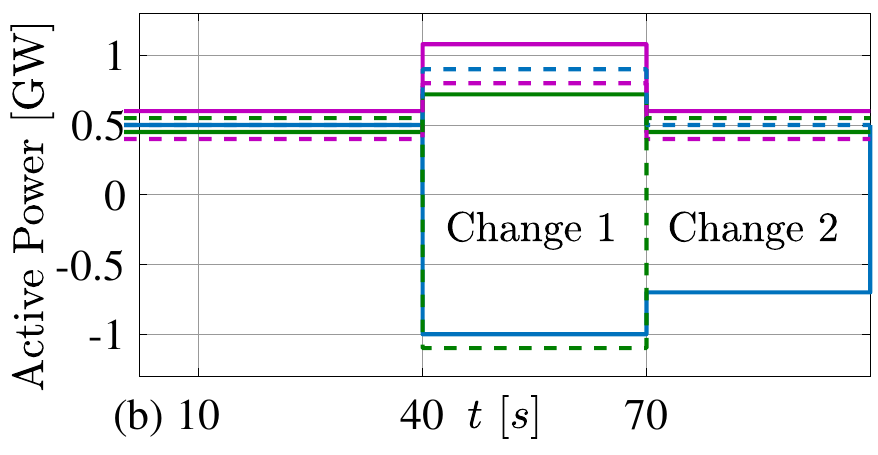}
\caption{Setpoints of the generated reactive and active powers on the ac-side of the offshore hub stations 7 to 12.\label{Fig:PowerRef}}
\end{figure}

\subsection{Case Study 1: Loss Reduction}
In this subsection, we apply the proposed controller to reduce the transmission losses of the HVdc grid. We show the system performance under the proposed controller \eqref{eq:PDD}-\eqref{eq:AlgebraicConstruction} considering the first cost function in Table~\ref{Table:Control}. The gradient of this loss function to be used in \eqref{eq:PDD} is $\nabla f(x_p,u_y)={\rm blkcol}(P_u u,0_6)$. We considered the generation-consumption scenario of offshore stations shown in Fig.~\ref{Fig:PowerRef} and plotted the results in Fig.~\ref{Fig:CaseStudy1}. Fig.~\ref{Fig:CaseStudy1}(a)-(b) show the voltage and current responses under the droop control, while Fig.~\ref{Fig:CaseStudy1}(c)-(d) show the results under the proposed controller activated at $t=10s$. We can see that both controllers maintain the voltage and current limits in steady state -- or more precisely, will converge to a steady state where the operating limits are met -- but they result in different values. As Fig.~\ref{Fig:CaseStudy1}(e) shows, the proposed online feedback optimization (OFO) controller steers the system to a steady state where the grid losses are smaller. This can be better understood from Fig.~\ref{Fig:CaseStudy1}(f). It can be seen that the proposed controller significantly reduces the grid losses compared to the droop control, and that this reduction depends on the generation-consumption condition.
\begin{figure}
\centering
\includegraphics[width=0.493\columnwidth]{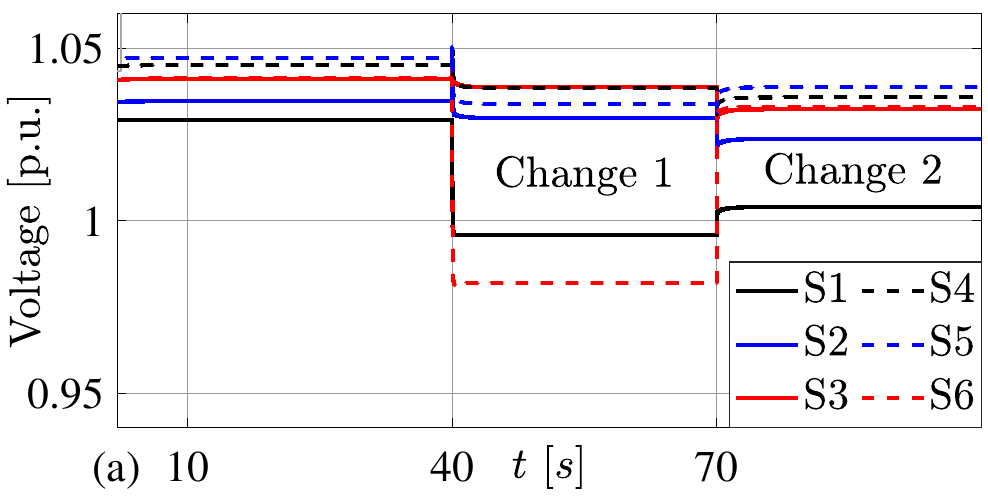}
\includegraphics[width=0.493\columnwidth]{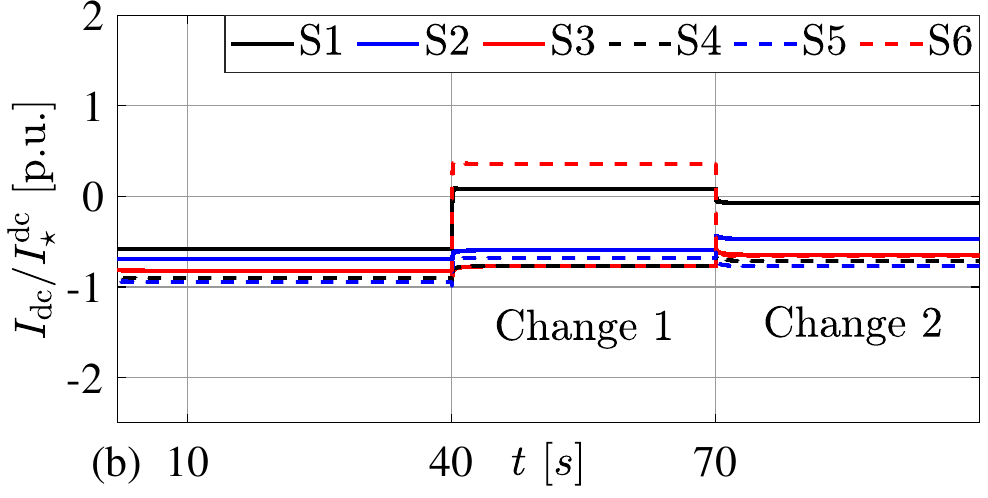}
\includegraphics[width=0.493\columnwidth]{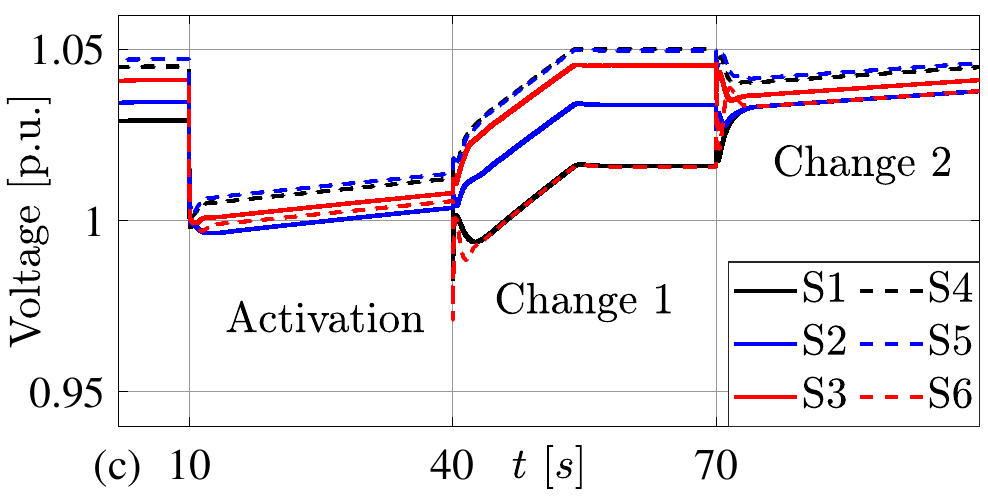}
\includegraphics[width=0.493\columnwidth]{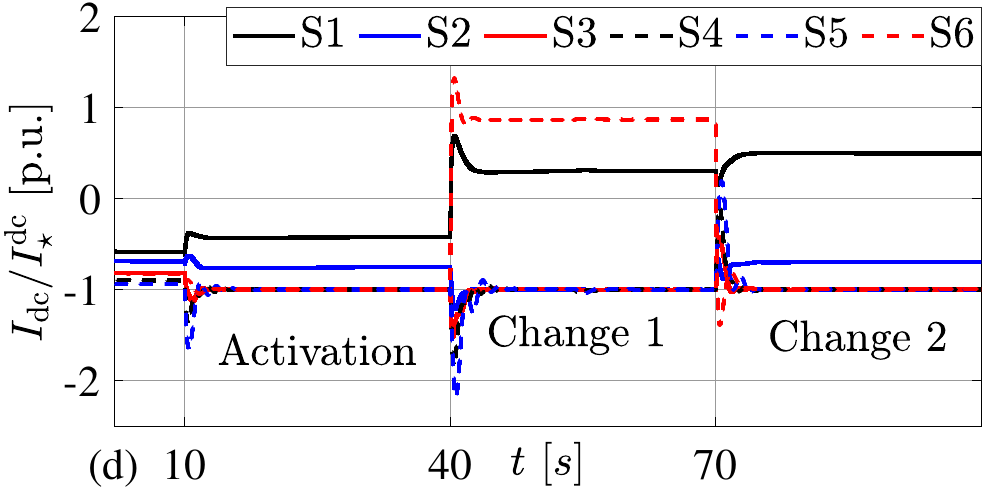}
\includegraphics[width=0.493\columnwidth]{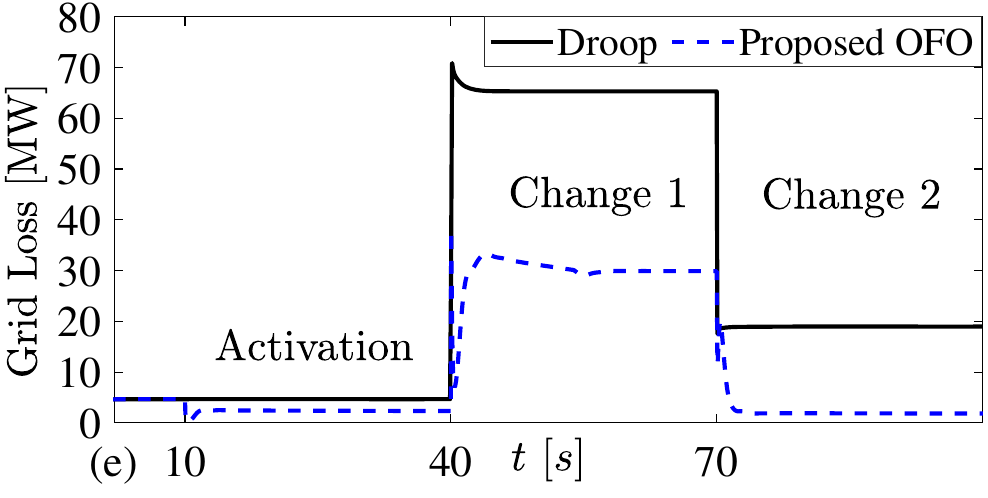}
\includegraphics[width=0.493\columnwidth]{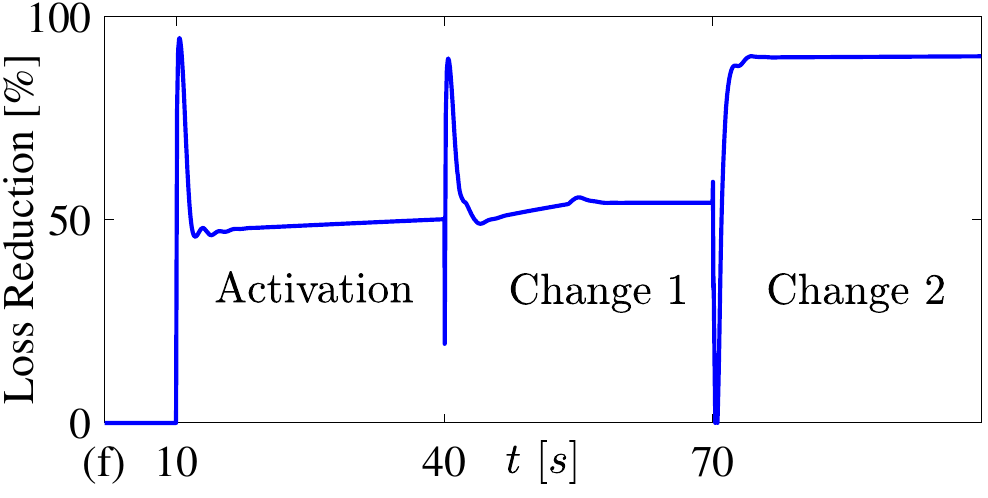}
\caption{Simulation results for Case Study 1; voltages $V_{{\rm dc}i}$ and current ratios $I_{{\rm dc}i}/I_{\star i}^{\rm dc}$ under (a)-(b) droop control and (c)-(d) proposed secondary (OFO) control; (e) grid loss under the two methods and (f) loss reduction under the proposed method compared to droop control.\label{Fig:CaseStudy1}}
\end{figure}


\subsection{Case Study 2: General Quadratic Output Optimization}
In this subsection, we show the system performance under the proposed controller \eqref{eq:PDD}-\eqref{eq:AlgebraicConstruction} considering the second cost function in Table~\ref{Table:Control}. The gradient of this function to be used in \eqref{eq:PDD} is $\nabla f(x_p,u_y)={\rm blkcol}(0_6,P_y u_y + q_y)$. We performed the same scenario as before and showed the results in Fig.~\ref{Fig:CaseStudy2}. Before $t=10s$, the system is controlled by the droop mechanism; we can see that the marginal costs are not equal. According to Fig.~\ref{Fig:CaseStudy2}(a)-(c), after activating the controller at $t=10s$, the voltages take a new formation such that the marginal costs of stations 1, 2 and 6 become equal, while the other stations absorb their maximum current. The new voltage formation is derived from the primal variables (Fig.~\ref{Fig:CaseStudy2}(d)), while the current limits of stations 3, 4, and 5 are maintained by the dual states $\zeta_i^{\rm min}$ (Fig.~\ref{Fig:CaseStudy2}(e)). At $t=40s$, after Change 1, the voltages take a new formation such that the marginal costs of the dc-GFM stations all become equal. In this condition, i.e. for $t\in[40s,70s]$, the principle of equal marginal costs is satisfied without violating the current limits (Fig.~\ref{Fig:CaseStudy2}(b) and Fig.~\ref{Fig:CaseStudy2}(e)). We can also see that the voltages are all within the limits thanks to the dual dynamics. In particular, from Fig.~\ref{Fig:CaseStudy2}(a) and Fig.~\ref{Fig:CaseStudy2}(f), we can see that, thanks to the dual state $\lambda_6^{\rm min}$, the voltage of the 6th station touches the minimum limit and stays there. At $t=70s$, the generation-consumption scenario is changed (Fig.~\ref{Fig:PowerRef}) and the controller reacts accordingly; the voltages take a new shape so that stations 1, 2, 3 and 6 get equal marginal costs, while every station respects the steady-state operating limits.

We have also shown the energy and frequency responses of the dc-GFM stations in Fig.~\ref{Fig:CaseStudy2}(g)-(h), validating our simulations and showing that the changes on the ac side of the offshore stations are properly reflected on the ac side of the area-connected stations. Fig.~\ref{Fig:CaseStudy2}(i) shows the online calculation of the equivalent conductance of the ac-GFM stations (cf. \eqref{eq:ACGFM-QS-Linear}) used to construct the grid input-output sensitivity matrix of the system (cf. \eqref{eq:AlgebraicConstruction}).

\begin{figure*}
\centering
\includegraphics[width=0.325\textwidth]{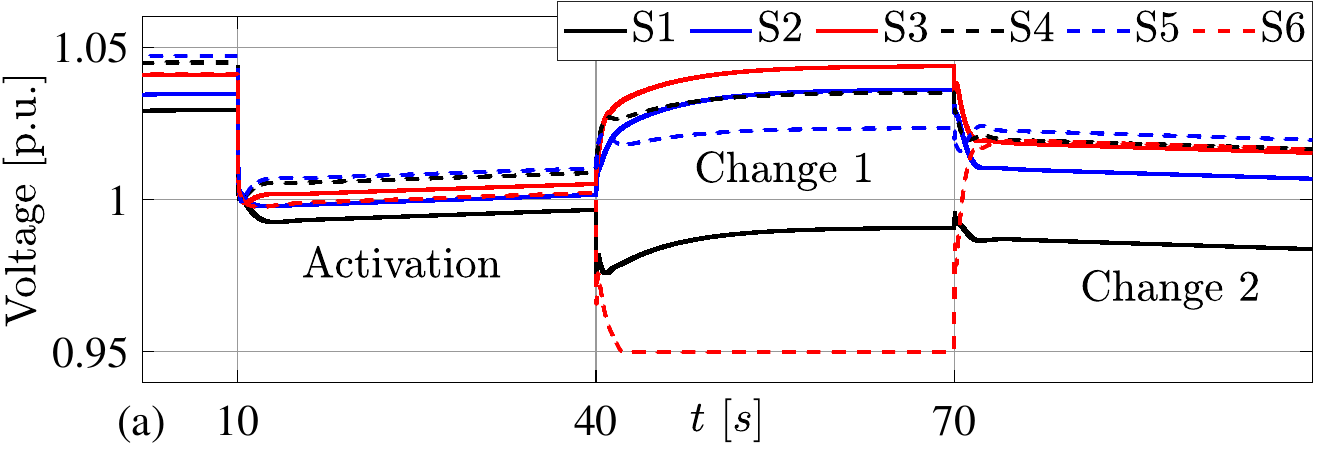}
\includegraphics[width=0.325\textwidth]{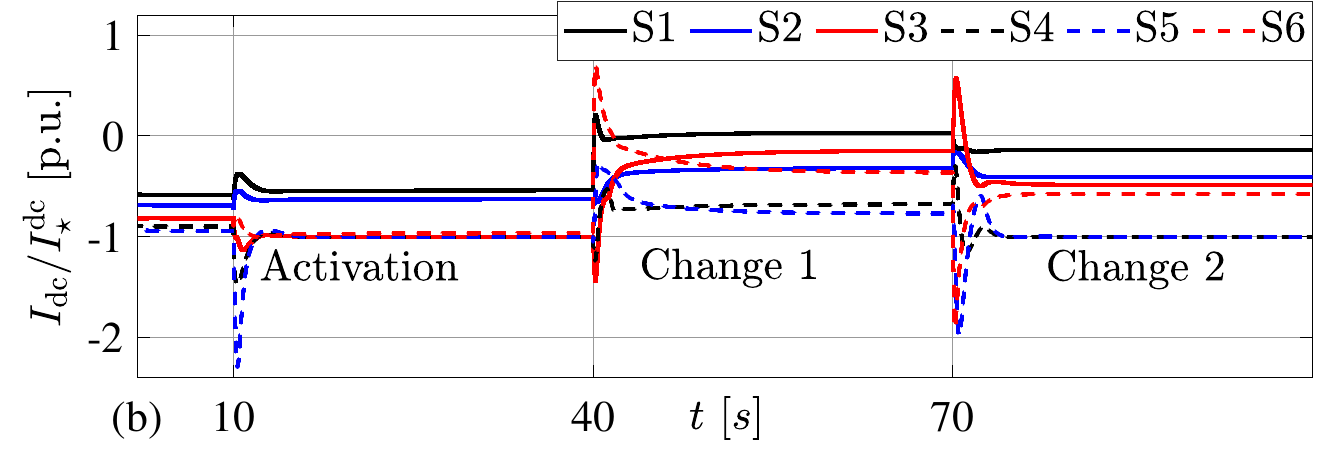}
\includegraphics[width=0.325\textwidth]{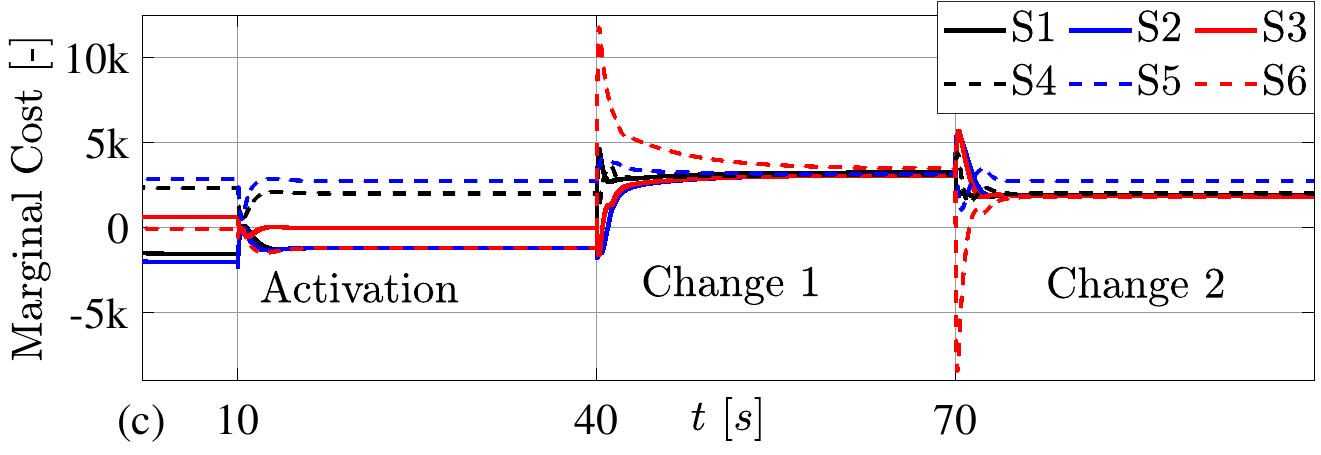}
\includegraphics[width=0.325\textwidth]{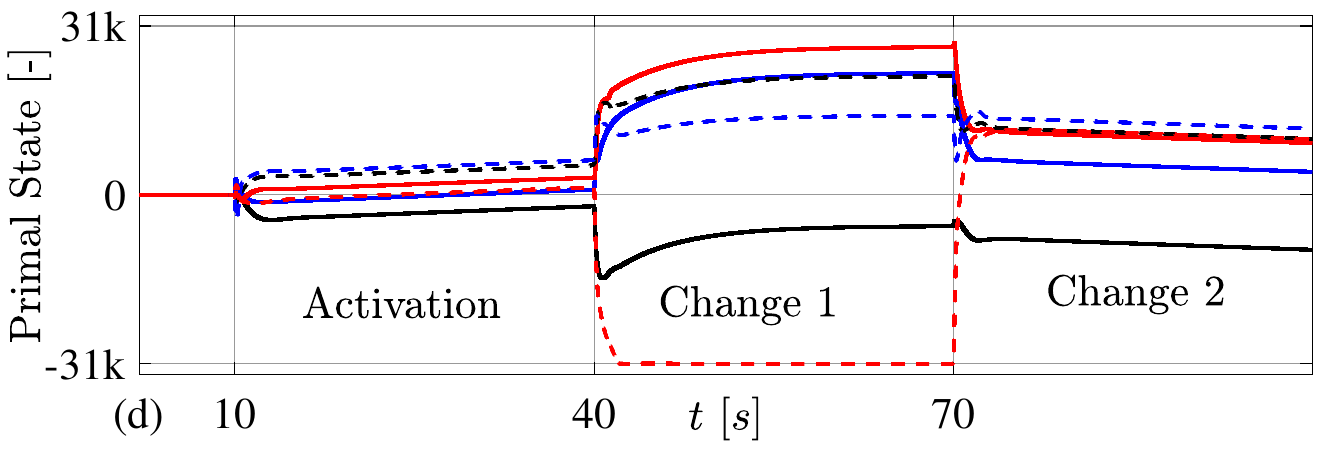}
\includegraphics[width=0.325\textwidth]{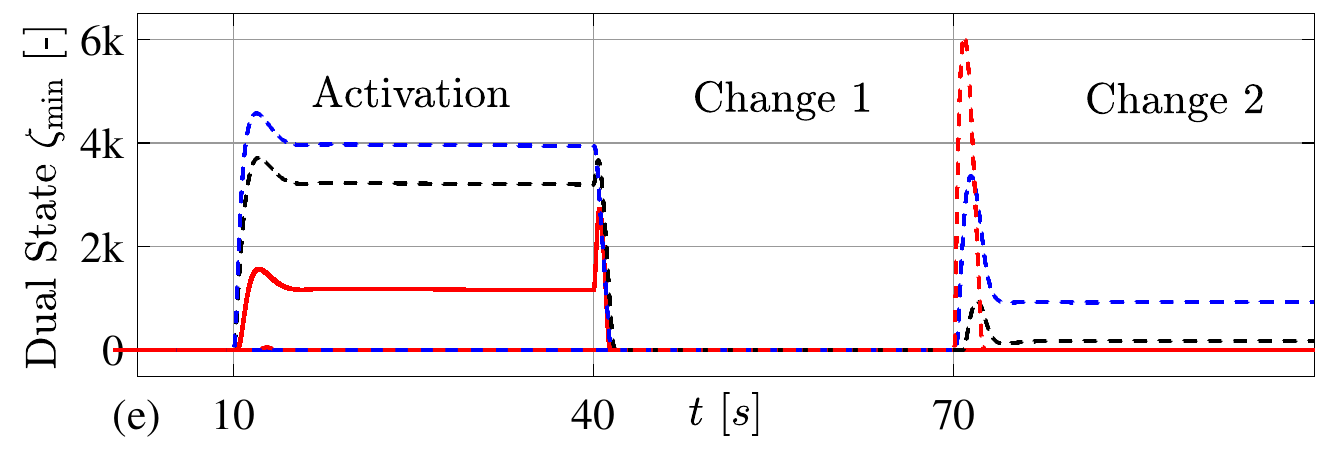}
\includegraphics[width=0.325\textwidth]{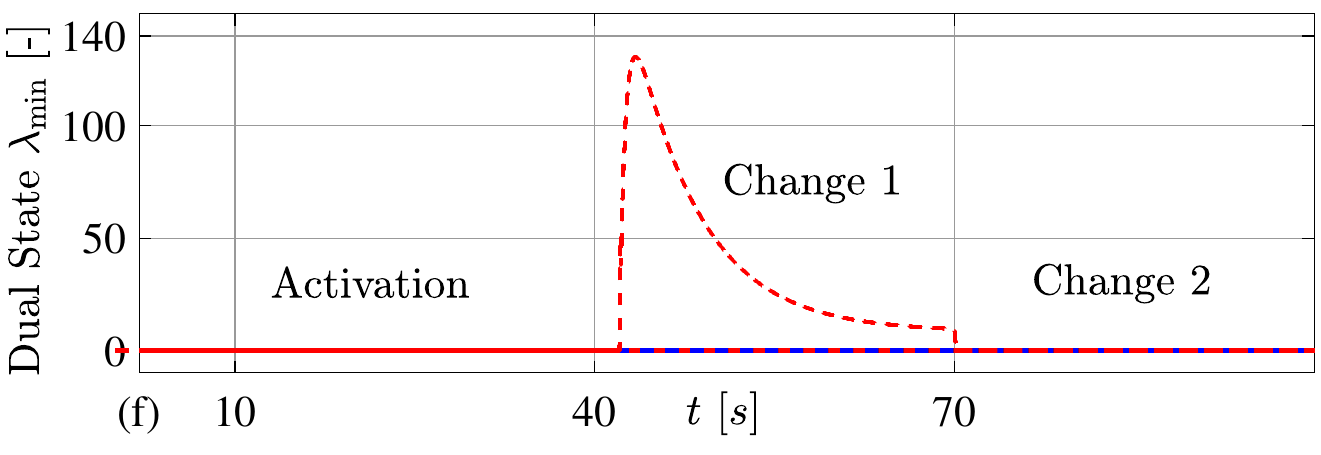}
\includegraphics[width=0.325\textwidth]{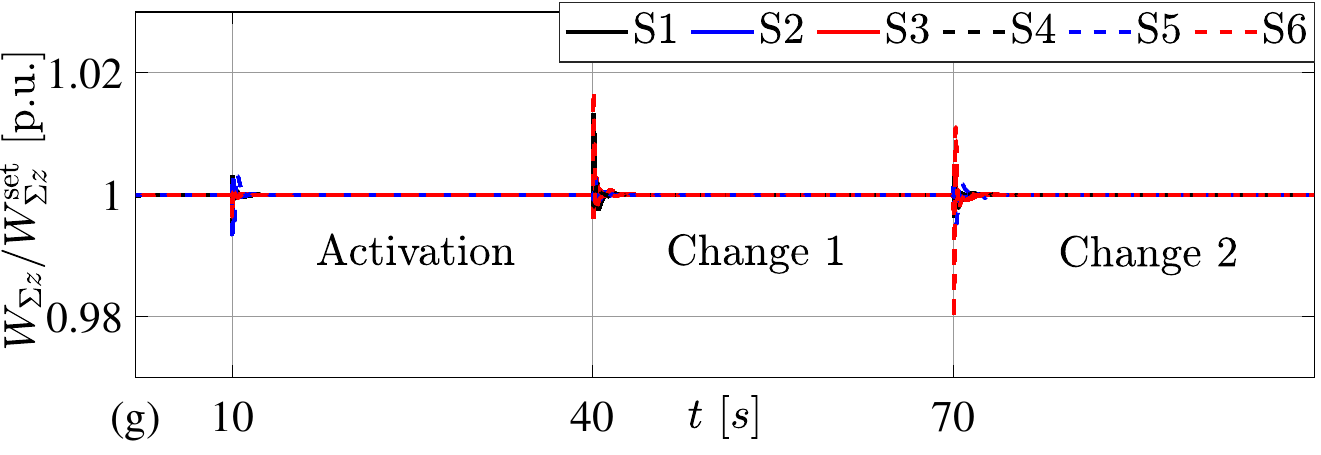}
\includegraphics[width=0.325\textwidth]{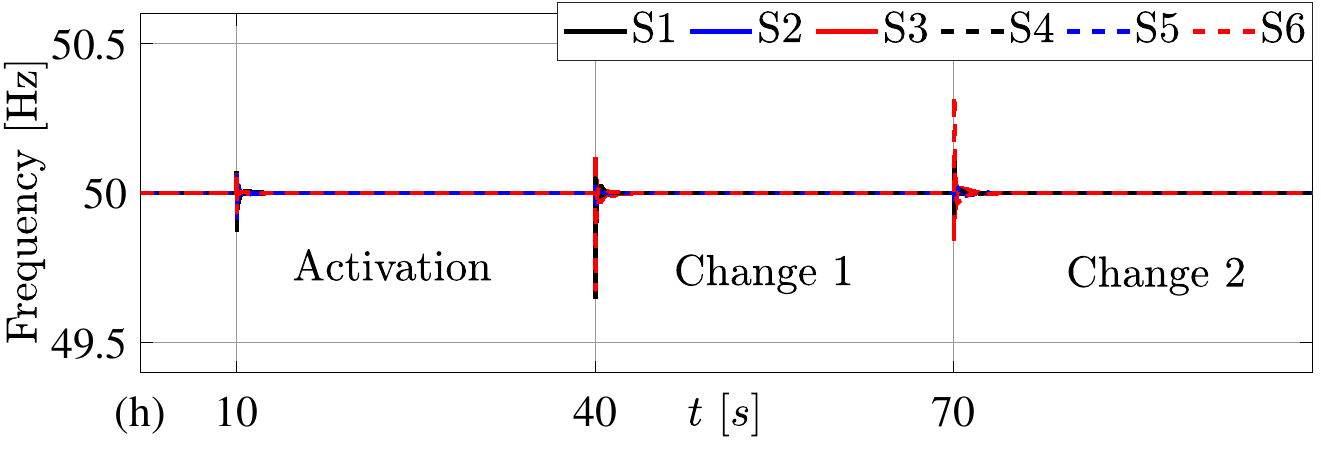}
\includegraphics[width=0.325\textwidth]{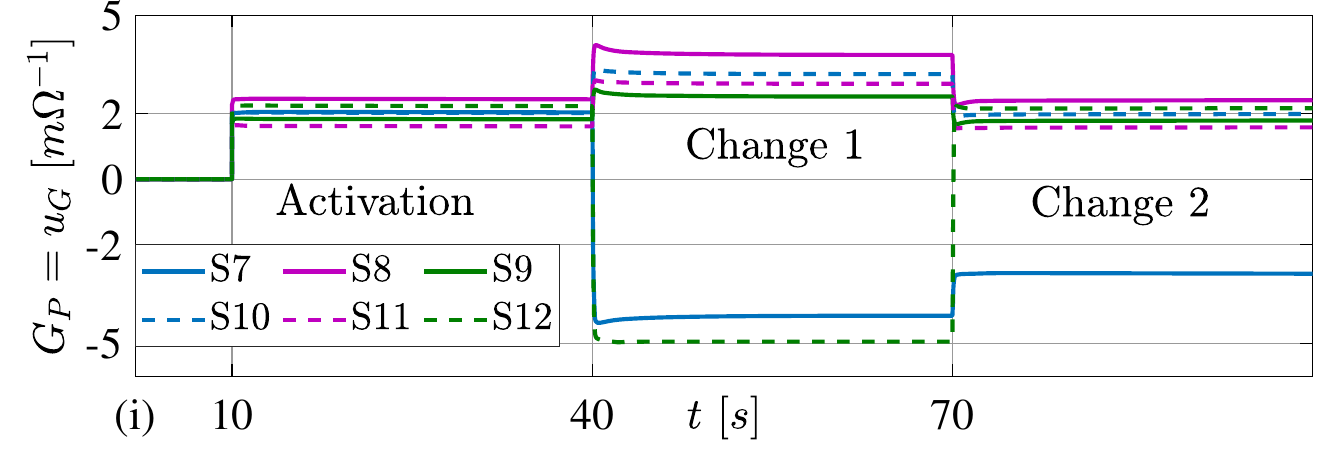}
\caption{Simulation results for Case Study 2; (a) voltages $V_{{\rm dc}i}$, (b) current ratios $I_{{\rm dc}i}/I_{\star i}^{\rm dc}$, (c) marginal costs (elements of $Py I_{\rm dc}+q_y$), (d) primal states $x_i^p$, (e) dual states $\zeta_i^{\rm min}$, (f) dual states $\lambda_i^{\rm min}$, (g) energy ratios $W_{{\Sigma z}i}/W_{{\Sigma z}i}^{\rm set}$, (h) frequencies \smash{$f_i=\tfrac{1}{2}\omega_i/\pi$}, and (i) ac-GFM conductances $G_{Pi}$.\label{Fig:CaseStudy2}}
\end{figure*}
\subsection{Case Study 3: Proportional Current Minimization Under Periodic vs. Reduced Communications}
In this case, we study the application of the proposed controller for proportional current minimization among the dc-GFM stations. The third cost function in Table~\ref{Table:Control} with the gradient $\nabla f(x_p,u_y)={\rm blkcol}(0_6,P_y u_y)$ is used in \eqref{eq:PDD}.
We also compare the results for both conventional periodic and the proposed event-based data transmission mechanisms.
To do this, we interconnect the controller and the system with \eqref{eq:ReducedInputs}, make use of the proposed scheme described in Fig.~\ref{Fig:CommScheme} and Algorithms \ref{LocalDCAlgorithm}-\ref{CentralAlgorithm}, and use the parameters in Table~\ref{Table:Control} for the triggering mechanisms in \eqref{eq:SamplingInstants}.
We ran the same scenario as in the previous case studies. The results under periodic (continuous) communication -- with constant frequency of 100 Hz -- are shown in the first column in Fig.~\ref{Fig:CaseStudy3}, while the second column shows the results under reduced communication. According to these results, the proposed controller can successfully minimize the station currents proportionally. Thus, the current demand is fairly distributed among the dc-GFM stations while respecting the operational limits. Furthermore, the results show almost identical performance under the two different communication mechanisms. However, according to Fig.~\ref{Fig:CommunicationNumber}, under the proposed reduced (event-based) communication mechanism, a significantly lower number of data sampling/transmissions is required. We have plotted the number of communications for different signals in Fig.~\ref{Fig:CommunicationNumber}. It can be seen that the transmission of the signals $\hat{G}_{Pi}$ happens less frequently compared to the signals $\hat{x}_i^p$ and $\hat{y}_i$; this is due to the low variations of $\hat{G}_{Pi}$ compared to the other states. 

\begin{figure}
\centering
\includegraphics[width=0.493\columnwidth]{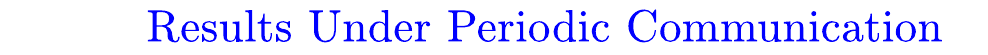}
\includegraphics[width=0.493\columnwidth]{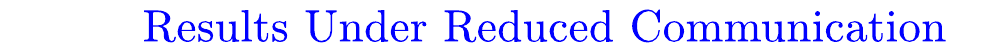}
\includegraphics[width=0.493\columnwidth]{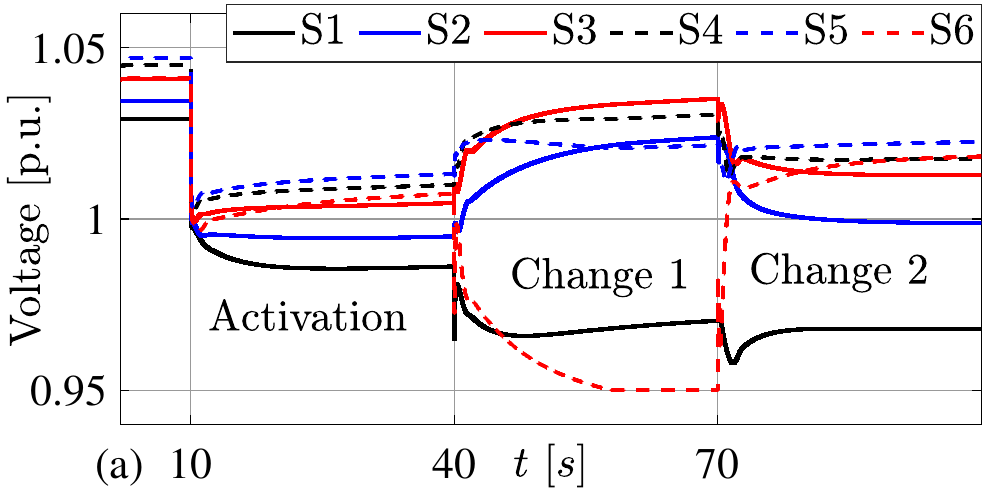}
\includegraphics[width=0.493\columnwidth]{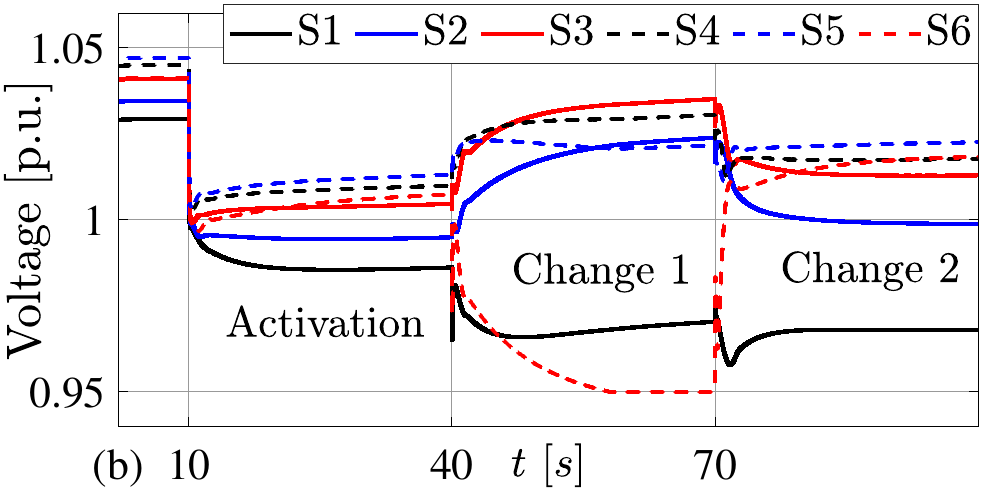}
\includegraphics[width=0.493\columnwidth]{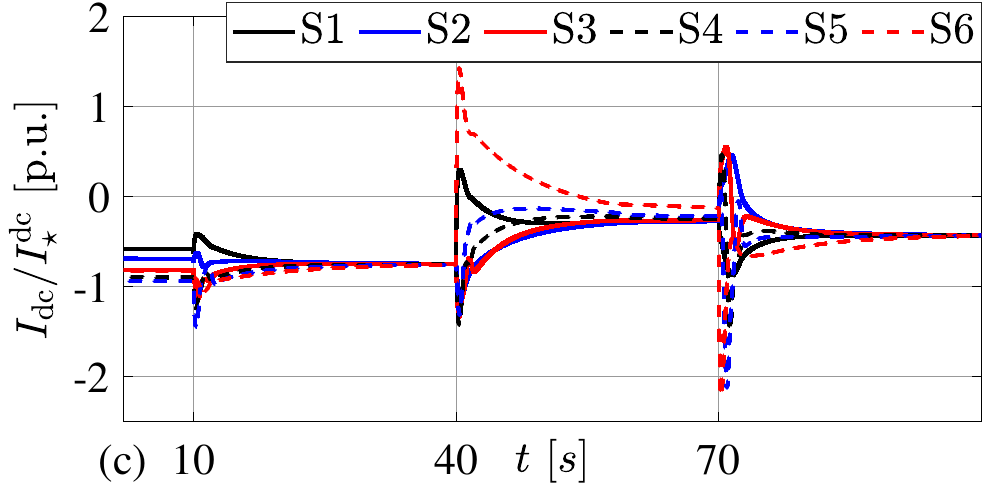}
\includegraphics[width=0.493\columnwidth]{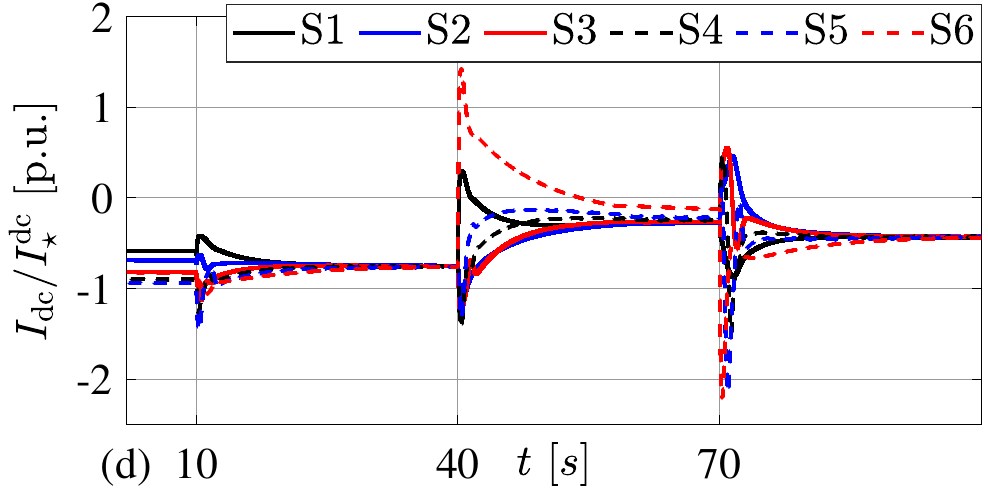}
\includegraphics[width=0.493\columnwidth]{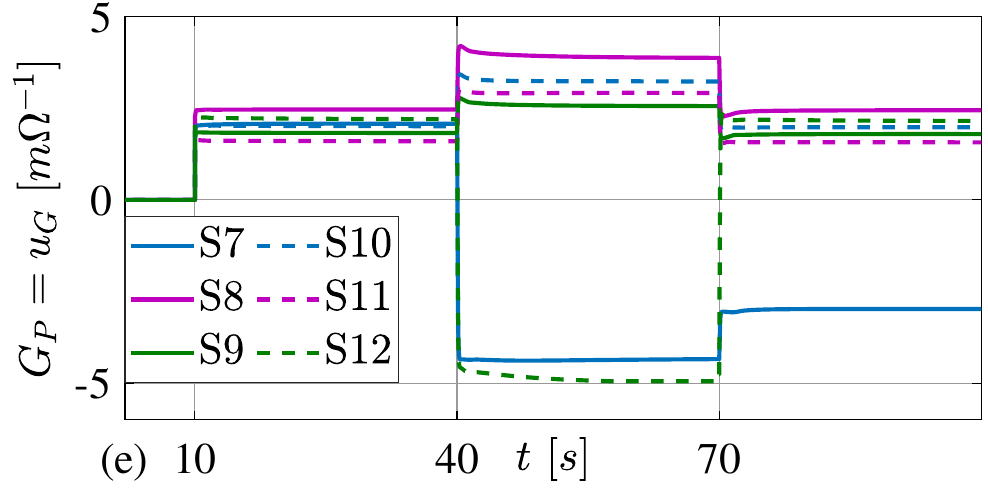}
\includegraphics[width=0.493\columnwidth]{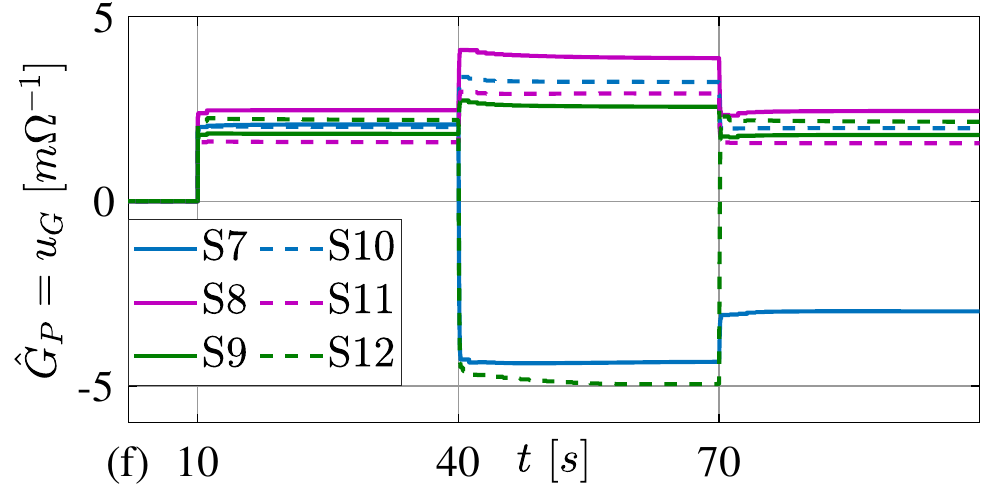}
\caption{Simulation results for Case Study 3; (a)-(b) voltages $V_{{\rm dc}i}$, (c)-(d) current ratios $I_{{\rm dc}i}/I_{\star i}^{\rm dc}$, and (e)-(f) ac-GFM conductances $G_i^P$. The first column shows the responses under constant frequency periodic communication, while the second column shows the results under reduced aperiodic communication.\label{Fig:CaseStudy3}}
\end{figure}
\begin{figure*}
    \centering
    \includegraphics[width=0.325\textwidth]{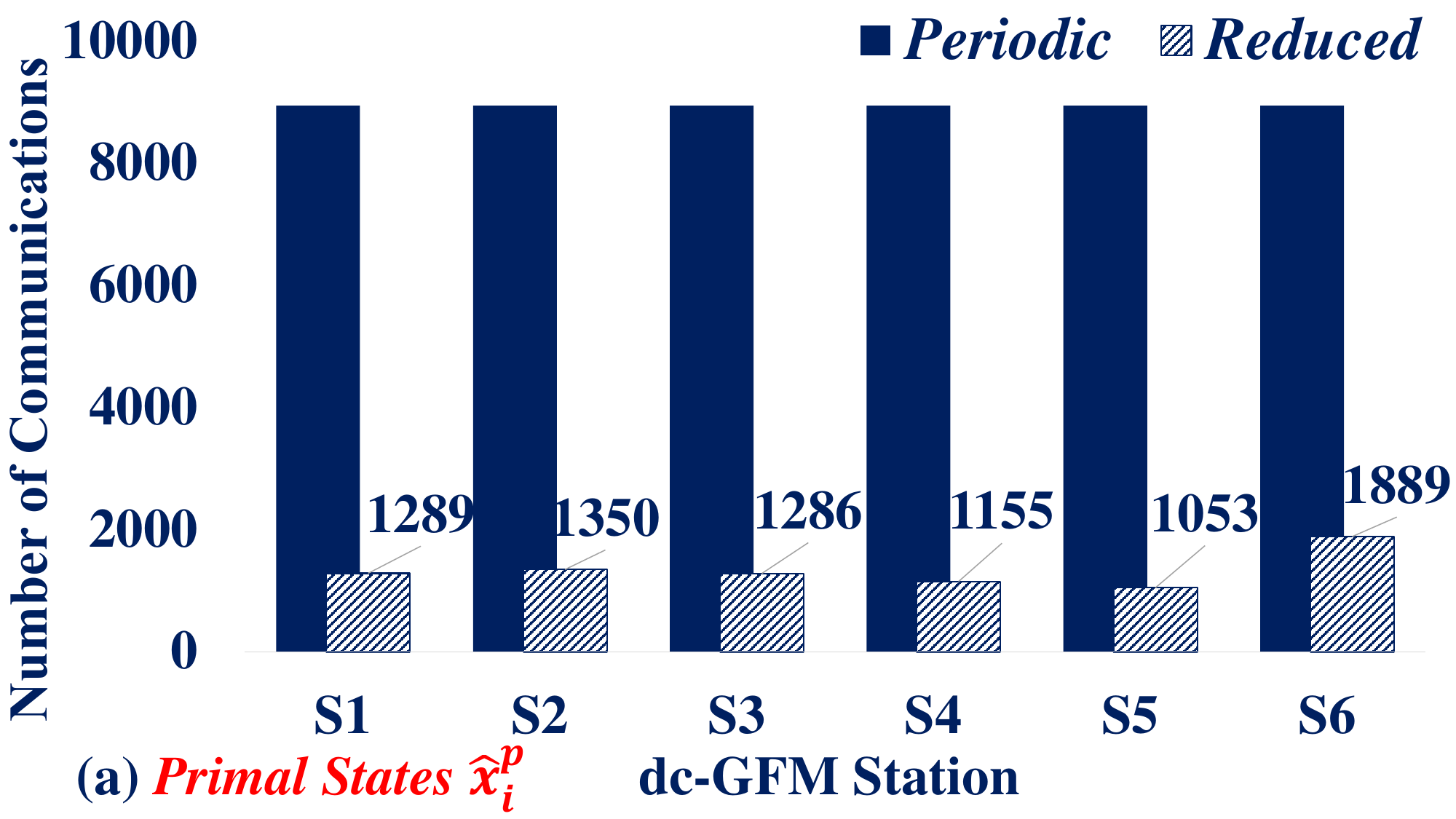}
    \includegraphics[width=0.325\textwidth]{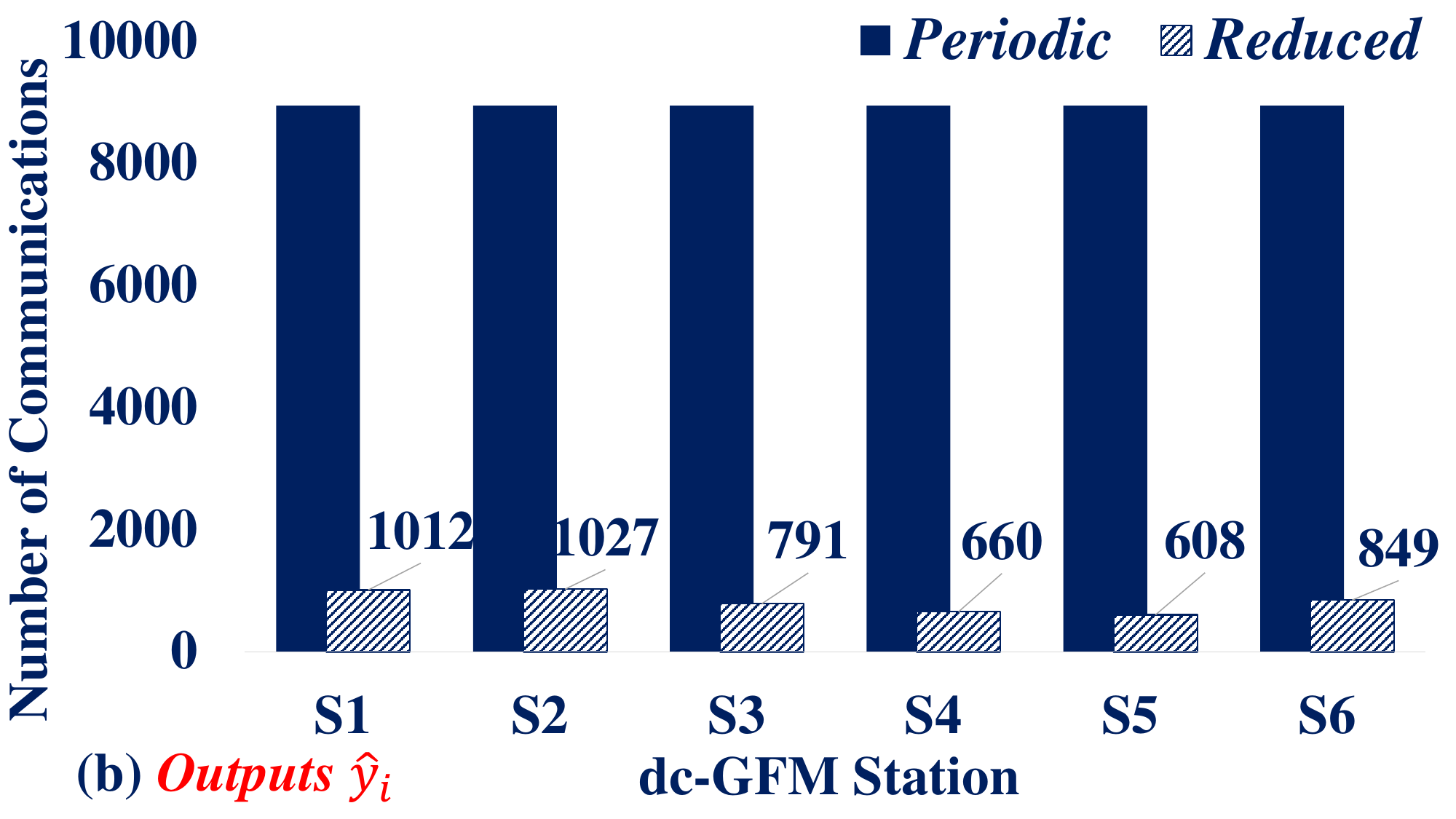}
    \includegraphics[width=0.325\textwidth]{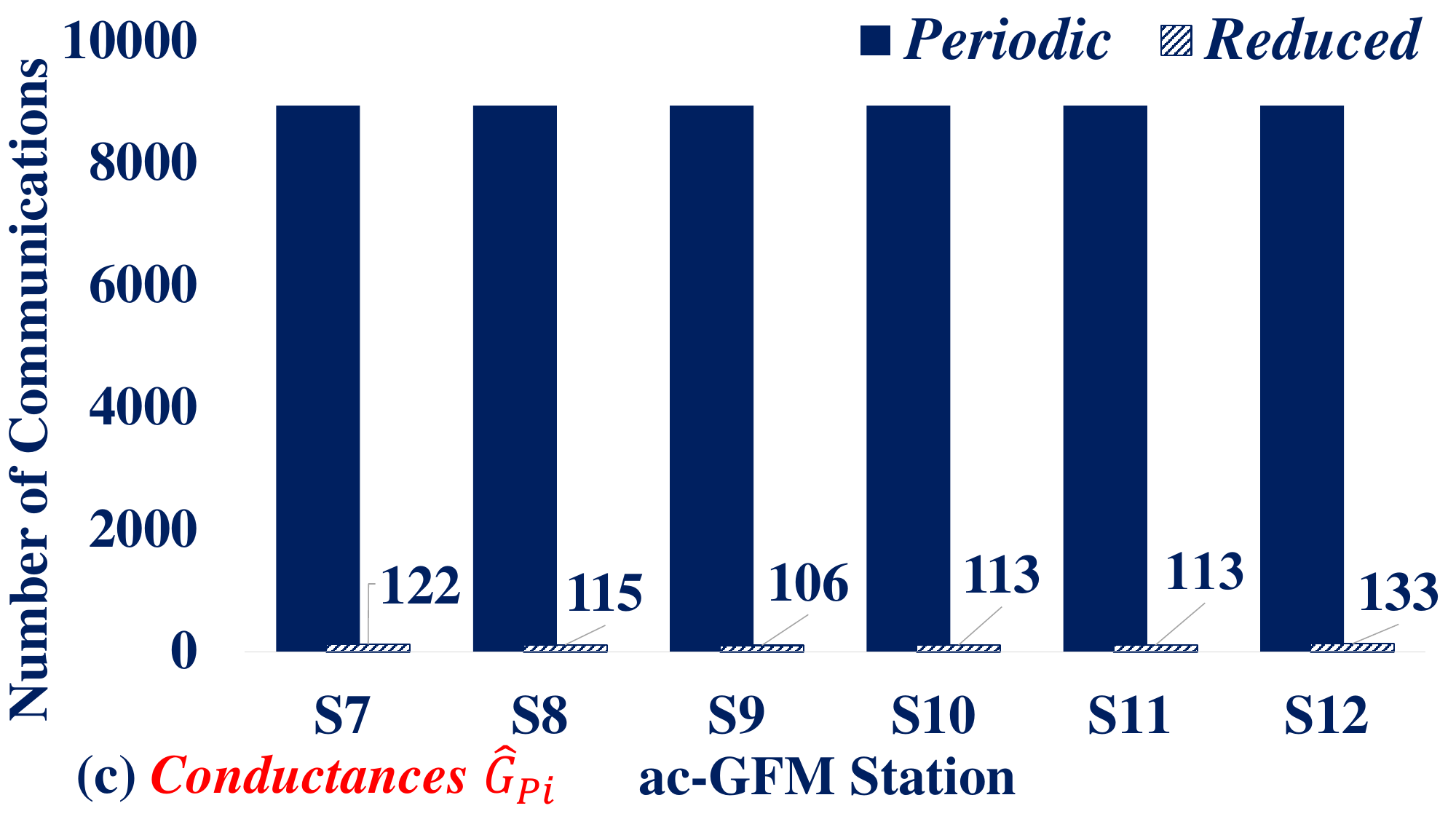}
    \caption{Number of communications (triggers) in Case Study 3 under periodic (100 Hz) vs. reduced communication mechanisms; the numbers are associated with the transmission of (a) primal states $\hat{x}_i^p$ from the central controller to the dc-GFM stations, (b) outputs $\hat{y}_i$ from the dc-GFM stations to the central controller, and (c) conductances $\hat{G}_{Pi}$ from the ac-GFM stations to the central controller.}
    \label{Fig:CommunicationNumber}
\end{figure*}


\section{Conclusion}
\label{Sec:Conclusion}

We have proposed a centralized secondary control for optimal steady-state operation of MT-HVdc grids under both voltage and current stationary limits. Towards this end, we first derived a suitable quasi-stationary model from a general dc transmission network  interconnecting dispatchable dc-GFM and non-dispatchable ac-GFM nodes. We then used this model to design an online feedback constraint optimization that steers the system towards optimality with stability guarantees. More precisely, we defined a general (convex) cost function which can adopt different objectives including loss reduction, economic dispatch and proportional current minimization. Furthermore, we showed that the inclusion of output (current) constraints in the optimization makes the controller rely upon the knowledge of the network model, albeit partially, and thus naturally favoring a centralized implementation. Our detailed simulation case studies based on an offshore MMC-based MT-HVdc grid have demonstrated the applicability of the proposed method for real-time optimization using different objective functions. Moreover, we have also shown that the communication traffic for implementing the proposed controller can be significantly reduced by using an event-based sampling mechanism without sacrificing the performance of the system. Finally, we acknowledge that potentially interesting research directions include replicating the results while adopting a distributed and model-free implementation instead.

\bibliographystyle{IEEEtran.bst}
\bibliography{IEEEabrv,References}

\end{document}